\documentclass[11pt,twoside]{article}

\newcommand{\otp}[3]{\ensuremath{[#1]\,#2\,[#3]}}

\newcommand{\pst}[1]{\ensuremath{\mathit{post}(#1)}}

\newcommand{\PrP}{\ensuremath{\alpha}\xspace}   
\newcommand{\PsP}{\ensuremath{\beta}\xspace}    

\newcommand{\Prg}{\ensuremath{P}\xspace}      

\newcommand{\PrPR}{\ensuremath{\alpha}\xspace}  
\newcommand{\PrR}{\ensuremath{P}\xspace}    
\newcommand{\PsPr}{\ensuremath{\beta}\xspace}    

\newcommand{\StO}{\ensuremath{\varphi}\xspace}   
\newcommand{\StT}{\ensuremath{\psi}\xspace}

\newcommand{\reg}[1]{\ensuremath{\mathtt{#1}}}   
\newcommand{\spec}[1]{\ensuremath{\red{\mathit{#1}}}}        

\newcommand{\PstPT}[1]{\ensuremath{R}\xspace}    

\newcommand{\asg}{\colonequals}
\newcommand{\nin}{\not\in}
\newcommand{\rng}[2]{\ensuremath{[#1:#2]}}

\newcommand{\gSt}{s} 
\newcommand{\inSt}{s} 
\newcommand{\fnSt}{t}     

\newcommand{\behs}[1]{\ensuremath{ \{\hspace{-0.2em}| #1 |\hspace{-0.2em}\} }}

\newcommand{\seqAx}{(\textit{Sequence Axiom})\xspace}

\newcommand{\emptPre}{(\textit{Empty Pre-program})\xspace}

\newcommand{\emptPrg}{(\textit{Empty Program})\xspace}

\newcommand{\trading}{(\textit{Trading})\xspace}

\newcommand{\appRule}{(\textit{Append})\xspace}

\newcommand{\eqSubRule}{(\textit{Substitution})\xspace}

\newcommand{\Peqv}{\equiv}  

\newcommand{\whileLaw}{(\textit{While})\xspace}

\newcommand{\whileLawC}{(\textit{While Consequence})\xspace}

\newcommand{\ifLaw}{(\textit{If})\xspace}

\newcommand{\ifLawO}{(\textit{One-way If})\xspace}

\newcommand{\seqLaw}{(\textit{Sequential Composition})\xspace}

\newcommand{\preRl}{(\textit{Pre-program Strengthening})\xspace}

\newcommand{\postRl}{(\textit{Post-program Weakening})\xspace}
\newcommand{\postRlNb}{\textit{Post-program Weakening}\xspace}

\newcommand{\eqNonR}{\textit{Equiv Nonrecursive}\xspace}
\newcommand{\eqNonRV}{\textit{Equiv Nonrecursive Void}\xspace}
\newcommand{\eqR}{\textit{Equiv Recursive}\xspace}

\newcommand{\PeqB}{\cong}    

\newcommand{\Pord}{\preceq}    
\newcommand{\PRef}{\sqsubseteq}    

\headheight     = \baselineskip
\headsep        = \baselineskip
\textheight     = \paperheight
\textwidth      = \paperwidth
\linewidth      = \textwidth

\input{TX/packages}

\newcommand{\bp}{\begin{proposition}}
\newcommand{\ep}{\end{proposition}}

\newcommand{\bco}{\begin{corollary}}
\newcommand{\eco}{\end{corollary}}

\newcommand{\bt}{\begin{theorem}}
\newcommand{\et}{\end{theorem}}

\newcommand{\bl}{\begin{lemma}}
\newcommand{\el}{\end{lemma}}

\newcommand{\bd}{\begin{definition}}
\newcommand{\ed}{\end{definition}}

\newcommand{\bex}{\begin{exmple}}
\newcommand{\eex}{\end{exmple}}

\newcommand{\bpr}{\vspace{-1.0ex} \begin{proof}}
\newcommand{\epr}{\end{proof}}

\newcommand{\bexer}{\begin{exercise}}
\newcommand{\eexer}{\end{exercise}}

\newcommand{\br}{\begin{remark}}
\newcommand{\er}{\end{remark}}

\newtheorem{theorem}{Theorem}
\newtheorem{lemma}[theorem]{Lemma}
\newtheorem{corollary}[theorem]{Corollary}
\newtheorem{proposition}[theorem]{Proposition}

\newtheorem{example}{Example}
\newenvironment{exmple}{\begin{example} \em}{\end{example}}

\newtheorem{exrcise}{Exercise}
\newenvironment{exercise}{\begin{exrcise} \em}{\end{exrcise}}

\newtheorem{remark}{Remark}
\newtheorem{definition}{Definition}

\newenvironment{proof}{\textbf{Proof.} } 
                      {\hfill{$\Box$}}

\def\squarebox#1{\hbox to #1{\hfill\vbox to #1{\vfill}}}

\newenvironment{proofsk}{\vspace{1ex}\textit{Proof sketch.}} {\qed}

\def\squarebox#1{\hbox to #1{\hfill\vbox to #1{\vfill}}}

\usepackage{calc}

\newcommand{\beqn}{\begin{centeqn}}
\newcommand{\eeqn}{\end{centeqn}}

\newcommand{\bleqn}[1]{\begin{centlabeqn}{#1}}
  \newcommand{\eleqn}{\end{centlabeqn}}

\newenvironment{centeqn}	
   {{\ss\\ \hspace*{\fill}}} 
   {\hspace*{\fill}\ss\\}

\newsavebox{\EqnLabel}
\newenvironment{centlabeqn}[1]	
   {\sbox{\EqnLabel}{#1}
    {\bs \hspace*{\fill}}
   } 
   {\hfill{\makebox[0in][r]{\usebox{\EqnLabel}}}\bs\\}

   {\sbox{\EqnLabel}{#1}
    {\bs \hspace*{\fill}}
   } 
   {\hfill{\makebox[0in][r]{\usebox{\EqnLabel}}}\bs\\}

\newenvironment{centeqn-nbsp} 
   {{\ms\\ \hspace*{\fill}}}
   {\hspace*{\fill}}

\newcommand{\be}{\begin{itemize}}
\newcommand{\ee}{\end{itemize}}
\newcommand{\bdn}{\begin{description}}
\newcommand{\edn}{\end{description}}
\newcommand{\bn}{\begin{enumerate}}
\newcommand{\en}{\end{enumerate}}

{\end{list}}

\newcommand{\remove}[1]{}

\newcommand{\bc}{\begin{center}}
\newcommand{\ec}{\end{center}}
\newcommand{\ul}{\underline}
\newcommand{\bs}{\bigskip}

\newcommand{\bfg}{\begin{figure}}
\newcommand{\efg}{\end{figure}}

\renewcommand{\ss}{\smallskip}

\newcommand{\eg}{e.g.,\xspace}
\newcommand{\ie}{i.e.,\xspace}
\newcommand{\wrt}{w.r.t.,\xspace}

\newcommand{\intr}{\empi}
\newcommand{\intrdef}{\emph}

\newcommand{\empi}[1]{\textit{#1}}

\newcommand{\dfn}[1]{\textbf{\textit{#1}}}

\newcommand{\ms}[1]{%
        \relax\ifmmode
                \mathord{\mathcode`\-="702D\it #1\mathcode`\-="2200}%
        \else
                $\mathord{\mathcode`\-="702D\it #1\mathcode`\-="2200}$%
        \fi
}

\newcommand{\ar}{\rightarrow}  

\renewcommand{\b}[1]{\overline{#1}}

\newcommand{\df}{\triangleq}

\renewcommand{\ge}{\geqslant}

\newcommand{\ints}{\cap}
\renewcommand{\l}{\ell}
\newcommand{\la}[1]{\mbox{$\, \stackrel{#1}{\rightarrow} \,$}}

\renewcommand{\le}{\leqslant}

\newcommand{\set}[1]{\ensuremath{\left\{ #1 \right\}}}

\newcommand{\stt}{\ensuremath{\ |\ }}
\newcommand{\sub}{\subseteq}

\newcommand{\un}{\cup}

\newcommand{\ang}[1]{\ifmmode{\left\langle #1 \right\rangle}
   \else{$\left\langle${#1}$\right\rangle$}\fi}

\newcommand{\sat}{\models}

\newcommand{\yld}{\vdash}

\newcommand{\true}{\ensuremath{\mathit{tt}}}
\newcommand{\false}{\ensuremath{\mathit{ff}}}

\newcommand{\ev}{\equiv}

\newcommand{\ex}{\ensuremath{\exists\,}}

\newcommand{\MIN}{\ensuremath{\mathrm{MIN}\,}}

\newcounter{lctr}
\newcommand{\reset}{\setcounter{lctr}{0}}

\newcommand{\lio}[1]{\addtocounter{lctr}{1}\>\ensuremath{\mathit{#1}}\\}
\newcommand{\lit}[1]{\addtocounter{lctr}{1}\>\>\ensuremath{\mathit{#1}}\\}
\newcommand{\lih}[1]{\addtocounter{lctr}{1}\>\>\>\ensuremath{#1}\\}

\newcommand{\lion}[1]{\addtocounter{lctr}{1}\>\ensuremath{#1}}
\newcommand{\litn}[1]{\addtocounter{lctr}{1}\>\>\ensuremath{#1}}

\newcommand{\pcomnt}[1]{\`// #1}
\newcommand{\lioc}[2]{\addtocounter{lctr}{1}\>\ensuremath{#1} \pcomnt{#2}\\}

\definecolor{dkgreen}{rgb}{0,0.6,0}
\definecolor{gray}{rgb}{0.5,0.5,0.5}
\definecolor{gray97}{gray}{.97}
\definecolor{gray45}{gray}{.45}
\definecolor{gray75}{gray}{.75}
\definecolor{mauve}{rgb}{0.58,0,0.82}
\definecolor{blueb}{rgb}{0,0,1}
\definecolor{redb}{rgb}{0.5,0,0}
\definecolor{darkblue}{rgb}{0.0,0.0,0.6}
\definecolor{cyan}{rgb}{0.0,0.6,0.6}

\newcommand{\red}[1]{\textcolor{red}{#1}}

\newcommand{\ttskp}{\ensuremath{\mathtt{skip}}}

\newcommand{\pseudocode}[1]{\ensuremath{{\boldsymbol{\mathit{#1}}}}}

\newcommand{\IFC}[1]{\pseudocode{if}\ (\ensuremath{#1})\ }
\newcommand{\IF}{\pseudocode{if}\ }
\newcommand{\THEN}{\pseudocode{then}\ }
\newcommand{\ELSE}{\pseudocode{else}\ }

\newcommand{\FI}{\pseudocode{fi}}

\newcommand{\WHILE}{\pseudocode{while}\ }
\newcommand{\WHILEC}[1]{\pseudocode{while}\ (\ensuremath{#1})\ }
\newcommand{\ELIHW}{\pseudocode{elihw}}

\newcommand{\FOR}{\pseudocode{for}\ }
\newcommand{\FORC}[1]{\pseudocode{for}\ (\ensuremath{#1})\ }
\newcommand{\ROF}{\pseudocode{rof}}
\newcommand{\TO}{\ \pseudocode{to}\ }
\newcommand{\DOWNTO}{\ \pseudocode{downto\ }}

\newcommand{\NODE}{\pseudocode{Node}}
\newcommand{\NEW}{\pseudocode{new}}
\newcommand{\RET}[1]{\pseudocode{return}(\ensuremath{#1})}

\newcommand{\ttcode}[1]{\ensuremath{\mathtt{{#1}}}}

\newcommand{\ttIFC}[1]{\ttcode{if}\ (\ensuremath{#1})}

\newcommand{\ttELSE}{\ttcode{else}\ }

\newcommand{\ttELSFC}[1]{\ttcode{else\ if}\ (\ensuremath{#1})\ }

\newcommand{\ttWHILEC}[1]{\ttcode{while}\ (\ensuremath{#1})\ }
\newcommand{\ttELIHW}{\ttcode{elihw}}

\newcommand{\ttFORC}[1]{\ttcode{for}\ (\ensuremath{#1})\ }
\newcommand{\ttROF}{\ \ttcode{rof}}

\newcommand{\ttNEW}{\ttcode{new}}
\newcommand{\ttRET}[1]{\ttcode{return}(\ensuremath{#1})}

\newcommand{\skp}{\ensuremath{\mathit{skip}}}   

\renewcommand{\sp}{\ |\ }

\newcommand{\Array}{\mathit{Array}}
\newcommand{\Arith}{\mathit{Arith}}
\newcommand{\Assgn}{\mathit{Assign}}
\newcommand{\Bool}{\mathit{Bool}}
\newcommand{\Class}{\mathit{Class}}
\newcommand{\ClassDef}{\mathit{ClassDef}}
\newcommand{\Num}{\mathit{Num}}
\newcommand{\Obj}{\mathit{Obj}}
\newcommand{\Expr}{\mathit{Expr}}
\newcommand{\EList}{\mathit{Expr\_List}}
\newcommand{\Ident}{\mathit{Id}}
\newcommand{\Iter}{\mathit{Iteration}}
\newcommand{\OCreate}{\mathit{Object\_Create}}

\newcommand{\Primitive}{\mathit{Primitive}}
\newcommand{\PList}{\mathit{Parameter\_List}}
\newcommand{\ProcList}{\mathit{Proc\_List}}
\newcommand{\ProcDef}{\mathit{Proc\_Def}}
\newcommand{\ProcCall}{\mathit{Proc\_Call}}

\newcommand{\VDec}{\mathit{Var\_Decl}}
\newcommand{\VList}{\mathit{Var\_List}}

\newcommand{\Stat}{\mathit{Stat}}
\newcommand{\Sel}{\mathit{Selection}}
\newcommand{\Type}{\mathit{Type}}

\newcommand{\Reference}{\mathit{Reference}}
\newcommand{\vod}{\mathit{void}}

\newcommand{\Ret}{\mathit{Ret}}
\newcommand{\Return}{\mathit{Return}}

\newcommand{\htp}[3]{\ensuremath{\{#1\}\,#2\,\{#3\}}}

\newcommand{\chns}{\ensuremath{[\hspace{-0.1ex}]}}
\newcommand{\ch}{\mbox{[\hspace{-0.1ex}]}}

\newcommand{\IDENT}[1]{\ensuremath{\mathit{#1}}\xspace}

\newcommand{\nil}{\mbox{\textsc{nil}}\xspace}

\newcommand{\key}{\IDENT{key}}

\newcommand{\Reals}{\mathit{\mathbf{R}}}

\renewcommand{\b}[1]{\overline{#1}}

\newlength{\tablesepp}
\begin{document}

\begin{titlepage}

\begin{center}

\vspace*{1.0in}

\Huge{Operational Annotations}

\Large{A new method for sequential program verification}

\vspace{0.5in}

\Large{Paul C. Attie}\\
School of Computer and Cyber Sciences\\
Augusta University
\vspace{1.0in}

\today
\vspace{0.75in}

\end{center}

\begin{abstract}

I present a new method for specifying and verifying the partial correctness of
sequential programs. The key observation is that, in Hoare logic, assertions are used as
selectors of states: an assertion specifies the set of program states which satisfy the assertion. 
Hence, the usual meaning of the partial correctness Hoare triple
\htp{f}{\Prg}{g}: if execution is started in \emph{any of the states that satisfy assertion $f$}, then,
upon termination, the resulting state will be \emph{some state that satisfies assertion $g$}.
There are of course other ways to specify a set of states.
I propose to replace assertions by terminating programs:
a program $\PrP$ specifies a set of states as follows: we start $\PrP$ in any state whatsoever,
and all the states that $\PrP$ may terminate in constitute the
specified set.
I call this set the \emph{post-states} of $\PrP$.
I introduce the \emph{operational triple}
\otp{\PrP}{\Prg}{\PsP} to mean: if execution of $\Prg$ is started in any post-state of $\PrP$, then upon termination,
the resulting state will be some post-state of $\PsP$.
Here, $\PrP$ is the \emph{pre-program}, and plays the role of a pre-condition, and
$\PsP$ is the \emph{post-program}, and plays the role of a post-condition.

\end{abstract}


\end{titlepage}


\section{Introduction}

I present a system for verifying partial correctness of deterministic
sequential programs. In contrast to Floyd-Hoare logic \cite{Fl67,Ho69}, my system does not
use pre-conditions and post-conditions, but rather
\intr{pre-programs} and \intr{post-programs}.
An assertion is essentially a means for defining a set of states:
those for which the assertion evaluates to true.
Hence the usual Hoare triple
\htp{f}{\Prg}{g} means that if execution of $\Prg$ is started in \emph{any of the
  states that satisfy assertion $f$}, then,
upon termination of $\Prg$, the resulting state will be \emph{some state that
  satisfies assertion $g$}.

Another method of defining a set of states is with a program \PrP which
starts execution in any state, \ie with precondition $\mathit{true}$.
The set of states in which \PrP terminates (taken over all possible
starting states) constitues the set of states that \PrP defines.
I call these the \intr{post-states} of \PrP.

I introduce the \intr{operational triple}
\otp{\PrP}{\Prg}{\PsP}, in which \PrP and \PsP are terminating sequential
programs, and $\Prg$ is the sequential program that is being verified.
\PrP is the \emph{pre-program} and
\PsP is the \emph{post-program}.
The meaning of \otp{\PrP}{\Prg}{\PsP} is as follows.
Consider executions of \PrP that start in any state 
(\ie any assignment of values to the variables).
From the final state of all such executions, \Prg is executed.
Let \StO be the set of resulting final states.
That is, \StO results from executing \Prg from any post-state of \PrP.
Also, let \StT be the set of post-states of \PsP, \ie the set of states
that result from executing \PsP starting in any state.
Then,  \otp{\PrP}{\Prg}{\PsP} is defined to mean $\StO \sub \StT$.
That is, the post-states of $\PrP;\Prg$ are a subset of the
post-states of $\PsP$.

The advantages of my approach are as follows.
Since the pre-program \PrP and the program \Prg are constituted from the same elements, namely program
statements, it is easy to ``trade'' between the two, \ie to move a statement from the program to the
pre-program and vice-versa. This tactic is put to good use in in
deriving programs from operational specifications, and is illustrated
in the examples given in this paper.
Since the pre-program \PrP and post-program \PsP are not actually executed,
then can be written without concern for efficiency. In fact, they can
refer to any well defined expression, \eg $\delta[t]$ for the
shortest path distance from a designated source $s$ to node $t$.


\remove{
My approach also removes the need to represent correctness properties properly in a declarative
logic. This can be helpful when representing such properties is
difficult, \eg for pointer-based programs, where normal Hoare logic does not
work, and more complex separation logic must be used. 
I illustrate this below with an example of verifying a program that reverses a singly-linked list in
place.
}

\section{Technical preliminaries}
\label{sec:prelim}

\subsection{Syntax of the programming language}

I use a basic programming language consisting of standard primitive types  (integers, boolean etc), arrays, and reference types,
assignments, if statements, while loops, for loops, procedure definition and invocation, class
definition, object creation and referencing.  The syntax that I use is given by the BNF grammar in
Table~\ref{table:syntax}.
I omit the definitions for $\Ident$ (identifier), $\Num$
(numeral), $\Obj$ (object reference), as these are primarily lexical in nature.  I assume as given
the grammar classes $\Primitive$ for primitive types, and $\Reference$ for reference types, \ie my
syntax is parametrized on these definitions. 

My syntax is standard and self-explanatory. I also use $\chns$ to denote
non-deterministic choice between two commands \cite{LawsP}.
For integers $i, j$ with $i \le j$, I use $x \asg \rng{i}{j}$ as
syntactic sugar for $x \asg i \ch  \cdots \ch x \asg j$, \ie a random
assignment of a value in $i,\ldots,j$ to $x$. This plays the role of
the range assertion $i \le x \le j$ in Hoare logic.
I use $\true$ for true, and $\false$ for false. 

In the sequel, when I use the term ``program''; I will mean a statement $(\Stat)$
written in the language of Table~\ref{table:syntax}.
I use $\true$ for true, and $\false$ for false.

\begin{table}
$\Type ::= \Primitive\ \sp\ \Reference\ \sp\ \vod$\\
  %
%
%
%
$\Array ::= \Ident[\Expr]$ \hfill //array element --- boolean or integer valued\\
%
%
$\Bool ::= \true \sp \false \sp \Ident \sp \Bool \land \Bool \sp \Bool \lor \Bool \sp
\neg \Bool \sp \Arith < \Arith \sp \Arith = \Arith$ \hfill //Boolean expression\\
$\Arith ::= \Ident \sp \Num \sp \Arith \sp \Arith + \Arith \sp \Arith * \Arith \sp \Arith - \Arith \sp - \Arith $ \hfill //integer-valued arithmetic expression\\
$\Expr ::= \Bool \sp \Arith \sp \Obj$ \hfill //expression\\
$\VDec ::= \Type\ \Ident$\\
$\VList ::= \VDec \sp \VDec, \VList$\\
$\EList ::= \Expr \sp \Expr, \EList$\\
$\ProcDef ::= \Type\ \Ident(\VList)\ \Stat$\\
$\ProcCall ::= \Ident(\PList) \sp \Ident \asg \Ident(\PList)$\\
$\ProcList ::= \ProcDef \sp \ProcDef\ \ProcList$\\
$\ClassDef ::= \Class\ \Ident\ \{ \VList;\ \ProcList \}$\\
$\OCreate ::= \Type\ \Obj := new\ \Ident(\PList)$ \hfill //object creation\\
$\Assgn ::= \Ident \asg \Expr \sp \Array \asg \Expr \sp \Obj \asg \Expr \sp \Obj.\Ident \asg  \Expr$\\
$\Sel ::=  \IF (\Bool)\ \THEN \Stat\ \ELSE \Stat\ \FI \sp \IF (\Bool) \ \THEN \Stat\ \FI$\\
$\Ret ::= \Return \sp  \Return(\Expr)$\\
$\Iter ::= \WHILEC{\Bool} \Stat\ \ELIHW \sp \FORC{\Ident = \Arith \TO \Arith} \Stat \ROF$\\
$\Stat ::= \skp \sp \Stat; \Stat \sp \Stat \,\ch\, \Stat \sp \Assgn \sp \Sel \sp \Iter \sp \ProcCall
\sp \OCreate$ \hfill //Statement 
\caption{Syntax of the programming language}
\label{table:syntax}
\end{table}

\subsection{Semantics of the programming language}

I assume the usual semantics for reference
types: object identifiers are pointers to the object, and the identity of an object is given by its
location in memory, so that two objects are identical iff they occupy the same memory.  Parameter
passing is by value, but as usual a passed array/object reference allows the called procedure to
manipulate the original array/object.
My proof method relies on (1) the axioms and inference rules introduced in this paper, and (2) an
underlying method for establishing program equivalence. Any semantics in which the above are valid
can be used. 
For concreteness, I assume a standard small-step (SOS) operational semantics \cite{Pl04,Sc14}.


An \dfn{execution} of program $\Prg$ is a finite sequence $\gSt_0, \gSt_1, \ldots, \gSt_n$ of states
such that (1) $\gSt_i$ results from a single small step of $\Prg$ in state $\gSt_{i-1}$, for all $i
\in 1,\ldots,n$, (2) $\gSt_0$ is an initial state of $\Prg$, and (3) $\gSt_n$ is a final (terminating) state of $\Prg$.
A \dfn{behavior} of program $P$ is a pair of states $(\inSt, \fnSt)$
such that (1) $\gSt_0, \gSt_1, \ldots, \gSt_n$ is an execution of $P$,
$\inSt = \gSt_0$, and $\fnSt = \gSt_n$.
Write $\behs{P}$ for the set of behaviors of $P$.

\subsection{Hoare Logic}

I assume standard first-order logic, with the standard model of
arithmetic and standard Tarskian semantics. Hence I take the notation
$\gSt \sat f$, where $\gSt$ is a state and $f$ is a formula, to have
the usual meaning.

I use the standard notation for Hoare-logic partial correctness: $\htp{f}{P}{g}$
means that if execution of program $P$ starts from a state satisfying
formula $f$, then, if $P$ terminates, the resulting final state will
satisfy formula $g$. That is, for all $(\inSt, \fnSt) \in \behs{P}$, if $\inSt \sat f$ then $\fnSt
\sat g$.


\section{Operational annotations}

I use a terminating program to specify a set of states. There is no
constraint on the initial states, and the set of all possible final states is specified.
If any initialization of variables is required, this must be done
explicitly by the program.

\begin{definition}[Post-state set]
\label{def:postSt}
Let  $\Prg$ be a program.
Then $$\pst{\Prg}  \df \set{ \fnSt \stt (\ex \inSt : (\inSt, \fnSt)  \in \behs{P}) }$$
\end{definition}

That is, $\pst{\Prg}$ is the set of all possible final states of $\Prg$, given any initial state. 
If one increases the set of states in which a program can start execution, then the set of states in
which the program terminates is also possibly increased, and is certainly not decreased. That is,
the set of post-states is monotonic in the set of pre-states.  Since prefixing a program $\PrP$ with
another program $\gamma$ simply restricts the states in which $\PrP$
starts execution, I have the following.

\begin{proposition}
\label{prop:ordering}
$\gamma; \PrP \Pord \PrP$.
\end{proposition}
\begin{proof}
Let $s \in \pst{\gamma;\PrP}$. Then, there is some state $u$ such that
$u \in \pst{\gamma}$ and $(u,s) \in \behs{\PrP}$. 
Hence $s \in \pst{\PrP}$. So $\pst{\gamma;\PrP} \sub \pst{\PrP}$, from which 
$\gamma; \PrP \Pord \PrP$ follows.
\end{proof}

The central definition of the paper is that of \intrdef{operational triple} \otp{\PrP}{\Prg}{\PsP}:

\begin{definition}[Operational triple]
\label{def:opTriple}
Let  $\PrP$, $\Prg$, and $\PsP$ be programs.  
Then
\[
  \otp{\PrP}{\Prg}{\PsP} \;\df\; \pst{\PrP; \Prg} \sub \pst{\PsP}.
\]
\end{definition}

\subsection{Program Ordering and Equivalence}

The next section presents a deductive system for establishing validity
of operational triples. The rules of inference use three kinds of
hypotheses: (1) operational triples (over ``substatements'' as usual),
and (2) program ordering assertions $P \Pord Q$, and (3) program
equivalence assertions $P \Peqv Q$.

\begin{definition}[Program ordering]
\label{def:progOrd}  
$P \Pord Q \df \pst{P}  \sub \pst{Q}$. 
\end{definition}
Here $P$ is ``stronger'' than $Q$ since it has fewer post-states (\wrt the
pre-condition \true, \ie all possible pre-states), and so $P$ produces
an output which satisfies, in general, more constraints than the
output of $Q$.
Note that this is \emph{not} the same as the usual program refinement
relation, since the mapping from pre-states to post-states induced by
the execution of $P$ is not considered.
Also, the ``direction'' of the inclusion relation is reversed
w.r.t. the usual refinement ordering, where we write $Q \PRef P$ to
denote that ``$P$ refines $Q$'', \ie $P$ satisfies more specifications
than $Q$. This is in keeping with the importance of the post-state set
in the sequel.
Note that, by Definition~\ref{def:progOrd} , 
$\otp{\PrP}{\Prg}{\PsP} \;\df\; {\PrP; \Prg} \Pord \pst{\PsP}.$

\begin{definition}[Program equivalence, $\Peqv$]
\label{def:progEquiv}
Programs $P$ and $Q$ are \emph{equivalent} iff they have the same behaviors$:$
$P \Peqv Q \df \behs{P} = \behs{Q}$.
\end{definition}
That is, I take as program equivalence the equality of program behaviors. 
Note that equivalence is \emph{not} ordering in both directions. 
This discrepancy is because ordering is used for weakening/strengthening laws (and
so post-state inclusion is sufficient) while equivalence is used for
substitution, and so, for programs at least, equality of behaviors
is needed.

Any method for establishing program ordering and equivalence is
sufficient for my needs. The ordering and equivalence
proofs in this paper were informal, and based on obvious concepts such
as the commutativity of assignment statements that modify different
variables/objects.

Works that provide proof systems for program equivalence include
\cite{CLRR16,PBD00,LR15,CG99,PS93}.  Some of these are mechanized, and
some use bisimulation and circular reasoning. I will look into using
these works for formally establishing program equivalence hypotheses
needed in my examples (which are then akin to Hoare logic verification
conditions), and to adapting these systems to establish program
ordering, \eg replace bisimulation by simulation.

\section{A deductive system for operational annotations}

\begin{table}
\begin{mathpar}
\otp{\PrP}{\Prg}{\PrP;\Prg} \hfill {\seqAx}\\

\otp{\skp}{\Prg}{\Prg}  \hfill {\emptPre}\\

\otp{\PrP}{\skp}{\PrP} \hfill {\emptPrg}\\

\otp{\PrP}{\Prg1; \Prg2}{\PsP} \;\mathrm{iff}\; \otp{\PrP;\Prg1}{\Prg2}{\PsP} \hfill {\trading}\\

\inferrule{ \otp{\PrP}{\Prg}{\PsP} }
              { \otp{\PrP}{\Prg;\gamma}{\PsP;\gamma} } 
                    \hfill  {\appRule}\\

\inferrule{ \PrP \Peqv \PrP'\;\; \Prg \Peqv \Prg'\;\; \PsP \Peqv \PsP'\;\;  \otp{\PrP}{\Prg}{\PsP} }
              { \otp{\PrP'}{\Prg'}{\PsP'} }
                         \hfill {\eqSubRule}\\

\inferrule{ \PrP \Pord \PrP' \;\;\; \otp{\PrP'}{\Prg}{\PsP} }
              { \otp{\PrP}{\Prg}{\PsP} } 
                        \hfill  {\preRl}\\

\inferrule{ \otp{\PrP}{\Prg}{\PsP'} \;\;\;  \PsP' \Pord \PsP }
              { \otp{\PrP}{\Prg}{\PsP} } 
                     \hfill  {\postRl}\\

\inferrule{ \otp{\PrP}{\Prg1}{\PsP}\;\;\;  \otp{\PsP}{\Prg2}{\gamma} }
              { \otp{\PrP}{\Prg1;\Prg2}{\gamma} } 
                   \hfill {\seqLaw}\\

\inferrule{ \otp{\PrP'}{\Prg}{\PrP} }
              { \otp{\PrP}{\WHILEC{B} \Prg\ \ELIHW}{\PsP} } 
                   \; {\PrP' \PeqB (\PrP, B), \PsP \PeqB (\PrP, \neg B)}
                          \hfill {\whileLaw}\\

\inferrule{ \otp{\PrP'}{\Prg}{\gamma;\PrP} }
              { \otp{\PrP}{\WHILEC{B} \Prg\ \ELIHW}{\PsP} } 
                   \; {\PrP' \PeqB (\PrP, B), \PsP \PeqB (\PrP, \neg B)}
                      \hfill {\whileLawC}\\
                      
\inferrule{ \otp{\PrP'}{\Prg1}{\PsP} \;\;\;\;\;  \otp{\PrP''}{\Prg2}{\PsP} }
              { \otp{\PrP}{\IF{B} \ \THEN  \Prg1 \ \ELSE \Prg2\ \FI}{\PsP} }
                    \; {\PrP' \PeqB (\PrP, B), \PrP'' \PeqB (\PrP, \neg B)} 
                          \hfill \ifLaw\\

\inferrule{ \otp{\PrP'}{\Prg}{\PsP} \;\;\;\;\;  \PrP'' \Pord \PsP }
              { \otp{\PrP}{\IF{B} \ \THEN  \Prg\ \FI}{\PsP} }
                   \; {\PrP' \PeqB (\PrP, B), \PrP'' \PeqB (\PrP, \neg B)} 
                          \hfill \ifLawO
\end{mathpar}  
\caption{Axioms and rules of inference}
\label{table:laws-opAnnMeaning}
\label{table:deductiveSystem}
\end{table}

Table \ref{table:deductiveSystem} presents a deductive system for operational annotations.
I do not provide a rule for the \FOR loop, since it can be easily 
turned into a \textbf{\textit{while}}. 
\FOR loops quite compact, and so are very convenient for use in 
pre-programs and post-programs, \ie \emph{as} annotations. 
The following are informal intuition and proofs of soundness for the axioms and inference
rules.


\subsection{Axioms}

\paragraph{Sequence Axiom.} 
If $\Prg$ executes after pre-program $\PrP$, the result is
identical to post-program $\PrP;\Prg$, \ie the sequential composition of $\PrP$ and $\Prg$. 
This gives an easy way to calculate a post-program for given
pre-program and program. The corresponding Hoare logic notion, namely
the strongest postcondition, is easy to compute (in closed
form) only for straight-line code.

\bp \seqAx is valid                      
\ep
\bpr
By Definition~\ref{def:opTriple}, $\otp{\PrP}{\Prg}{\PrP;\Prg}$ is
$\pst{\PrP; \Prg} \sub \pst{\PrP; \Prg}$, which is immediate.
\epr

\paragraph{Empty pre-program.} 
Follows from \seqAx by replacing  $\PrP$ by the empty program \skp.
Program \Prg is ``doing all the work'', and so the resulting
post-program is also \Prg. Having \skp\ as a pre-program is similar to
having \true\ as a precondition in Hoare logic.

\paragraph{Empty Program.} 
Follows from \seqAx by replacing  $\Prg$ by the empty program \skp.
This is analogous to the axiom for \skp\ in Hoare logic:
$\htp{f}{\skp}{f}$, since \skp\ has no effect on the program state.

\subsection{Rules of Inference}

\paragraph{Trading Rule.} 
Sequential composition is associative: $\PrP; (\Prg1; \Prg2) \Peqv (\PrP; \Prg1); \Prg2$.
By Definitions~\ref{def:progEquiv}, \ref{def:postSt}:
$\pst{\PrP; (\Prg1; \Prg2)} = \pst{(\PrP; \Prg1); \Prg2}$
Hence, if the program is a sequential composition $\Prg1; \Prg2$,
I can take $\Prg1$ and add it to the end of the pre-program $\PrP$.
I can also go in the reverse direction, so technically there are two rules of inference here. I will refer to
both rules as \trading.
This \emph{seamless transfer between program and pre-program has no analogue in
Hoare logic}, and provides a major tactic for the derivation of
programs from operational specifications.

\bp \trading is sound, that is, each side holds iff the other does.
\ep
\bpr
By Definition~\ref{def:opTriple},
$\otp{\PrP}{\Prg1; \Prg2}{\PsP}$ is $\pst{\PrP; \Prg1; \Prg2} \sub \pst{\PsP}$, and
$\otp{\PrP; \Prg1}{\Prg2}{\PsP}$ is also $\pst{\PrP; \Prg1; \Prg2} \sub \pst{\PsP}$.
Hence each holds iff the other does.
\epr

\paragraph{Append Rule.} 
Appending the same program $\gamma$ to the program and the
post-program preserves the validity of an operational triple.
This is useful for appending new code into both the program and the post-program.

\bp \appRule is sound, that is, if the hypothesis holds, then so does the conclusion.
\ep
\bpr
I must show $\otp{\PrP}{\Prg;\gamma}{\PsP;\gamma}$, which by Definition~\ref{def:opTriple} is 
$\pst{\PrP; \Prg; \gamma} \sub \pst{\PsP; \gamma}$.
Let $s \in \pst{\PrP; \Prg; \gamma}$. Hence there is some state $t$ such that
$t \in \pst{\PrP; \Prg}$ and $(t,s) \in \behs{\gamma}$.
By assumption, $\otp{\PrP}{\Prg}{\PsP}$, which by Definition~\ref{def:opTriple} is $\pst{\PrP; \Prg} \sub \pst{\PsP}$.
Since $t \in \pst{\PrP; \Prg}$, I have $t \in \pst{\PsP}$.
Since $(t,s) \in \behs{\gamma}$, I also have $s \in \pst{\PsP; \gamma}$.
Since $s$ was chosen arbitrarily, I conclude
$\pst{\PrP; \Prg; \gamma} \sub \pst{\PsP; \gamma}$.
\epr

\paragraph{Substitution Rule.} 
Since the definition of operational annotation refers only to the
behavior of a program, it follows that 
one equivalent program can be replaced by another, 
This rule is useful for performing equivalence-preserving transformations, such as loop unwinding.

\bp \eqSubRule is sound, that is, if the hypothesis holds, then so does the conclusion.
\ep
\bpr
$\PrP \Peqv \PrP'$ implies that $\pst{\PrP} = \pst{\PrP'}$.
$\Prg \Peqv \Prg'$ means that $\behs{\Prg} = \behs{\Prg'}$.
Hence $\pst{\PrP;\Prg} = \pst{\PrP';\Prg'}$.
$\PsP \Peqv \PsP'$ implies that $\pst{\PsP} = \pst{\PsP'}$.
Hence  $\pst{\PrP;\Prg} \sub \pst{\PsP}$ iff $\pst{\PrP';\Prg'} \sub \pst{\PsP'}$.
Hence  $\otp{\PrP}{\Prg}{\PsP}$ iff $\otp{\PrP'}{\Prg'}{\PsP'}$.
\epr

\subsection{Rules of inference that are analogues of Hoare logic laws}

I now present syntax-based rules, which are straightforward analogues
of the corresponding Hoare logic rules.

\paragraph{Pre-program strengthening.}
Reducing the set of post-states of the pre-program cannot invalidate
an operational triple.

\bp \preRl is sound. 
\ep
\bpr
Let $t \in \pst{\PrP; \Prg}$. Then there exists a state $s$ such that
$s \in \pst{\PrP}$ and $(s,t) \in \behs{\Prg}$.
By assumption, $\PrP \Pord \PrP'$, and so $\pst{\PrP} \sub \pst{\PrP'}$.
Hence $s \in \pst{\PrP'}$.
From this and  $(s,t) \in \behs{\Prg}$,
I conclude  $t \in \pst{\PrP'; \Prg}$. 
Since $t$ is arbitrarily chosen, I have $\pst{\PrP; \Prg} \sub \pst{\PrP'; \Prg}$.
Hypothesis $\otp{\PrP'}{\Prg}{\PsP}$ means $\pst{\PrP'; \Prg} \sub \pst{\PsP}$.
Hence $\pst{\PrP; \Prg} \sub \pst{\PsP}$, 
and so $\otp{\PrP}{\Prg}{\PsP}$ by Definition~\ref{def:opTriple}.
\epr

\paragraph{Post-program weakening.}
Enlarging the set of post-states of the post-program cannot invalidate
an operational triple.

\bp \postRl is sound.
\ep
\bpr
Hypothesis $\otp{\PrP}{\Prg}{\PsP'}$ means $\pst{\PrP; \Prg} \sub \pst{\PsP'}$. 
Hypothesis $\PsP' \Pord \PsP$ means that $\pst{\PsP'} \sub \pst{\PsP}$. 
Hence $\pst{\PrP; \Prg} \sub \pst{\PsP}$, and so $\otp{\PrP}{\Prg}{\PsP}$
by Definition~\ref{def:opTriple}.
\epr

\paragraph{Sequential composition.}
The post-state set of $\PsP$ serves as the intermediate state-set in the execution of $\Prg1;\Prg2$:
it characterizes the possible states after $\Prg1$ executes and before $\Prg2$ executes.
That is, execution of \Prg1 starting from a post-state of $\PrP$ yields a post-state of \PsP.
Then execution of \Prg2 starting from a post-state of $\PsP$ yields a post-state of $\gamma$.

\bp \seqLaw is sound.
\ep
\bpr
Let $t \in \pst{\PrP;\Prg1;\Prg2}$. Then there is some $s$ such that
$s \in  \pst{\PrP;\Prg1}$ and $(s,t) \in \behs{\Prg2}$.
From hypothesis $\otp{\PrP}{\Prg1}{\PsP}$, I have $\pst{\PrP;\Prg1} \sub \pst{\PsP}$.
Hence $s \in \pst{\PsP}$. From this and $(s,t) \in \behs{\Prg2}$, I have
$t \in \pst{\PsP; \Prg2}$.
From hypothesis $\otp{\PsP}{\Prg2}{\gamma}$, I have $\pst{\PsP;\Prg2} \sub \pst{\gamma}$.
From this and $t \in \pst{\PsP; \Prg2}$, I have $t \in \pst{\gamma}$.
Since $t$ is chosen arbitrarily, I conclude $\pst{\PrP;\Prg1;\Prg2} \sub
\pst{\gamma}$.
Hence $\otp{\PrP}{\Prg1;\Prg2}{\gamma}$ by Definition~\ref{def:opTriple}.
\epr

\paragraph{While Rule.} 
Given \otp{\PrP}{\Prg}{\PrP}, I wish to conclude \otp{\PrP}{\WHILEC{B}\Prg}{\PsP} where 
\PsP is a ``conjunction'' of \PrP and $\neg B$, \ie the post-states of
\PsP are those that are post-states of \PrP, and also that satisfy
assertion $\neg B$, the negation of the looping condition.
Also, I wish to weaken the hypothesis of the rule from \otp{\PrP}{\Prg}{\PrP} to
\otp{\PrP'}{\Prg}{\PrP}, where $\PrP'$ is a ``conjunction'' of \PrP
and $B$, \ie
the post-states of $\PrP'$ are those that are post-states of $\PrP$ and
that also satisfy assertion $B$, the looping condition.
I therefore define the ``conjunction'' of a program and an assertion as follows.

\bd[Conjunction of program and condition]
Let $\PrP', \PrP$ be programs and $B$ a Boolean expression. Then define
$$\PrP' \PeqB (\PrP, B) \df \pst{\PrP'} = \pst{\PrP} \ints \set{ s \stt s(B) = \true }.$$
\ed
Note that this 
definition does not produce a unique result, and so is really a relation rather than a mapping.
The construction of $\PrP'$ is not straightforward, in general, for arbitrary assertion $B$.
Fortunately, most looping conditions are simple, typically a loop counter reaching a
limit. I therefore define the needed program $\PrP'$ by the semantic condition given above, and
leave the problem of deriving $\PrP'$ from $\PrP$ and $B$ to another occasion.

Given a while loop $\WHILEC{B} \Prg\ \ELIHW$ and pre-program $\PrP$,
let $\PrP'$ be a program such that $\PrP' \PeqB (\PrP, B)$, and
let $\PsP$ be a program such that $\PsP \PeqB (\PrP, \neg B)$.
The hypothesis of the rule is: execute $\PrP$ and restrict the set of
post-states to those in which $B$ holds. That is,
have $\PrP'$ as a pre-program for the loop body $\Prg$.
Then, after $\Prg$ is executed, the total resulting effect must be the
same as executing just $\PrP$. So, $\PrP$ is a kind of ``operational invariant''.
Given that this holds, and taking $\PrP$ as a pre-program for 
$\WHILEC{B} \Prg\ \ELIHW$, then upon termination, 
we have $\PrP$ as a post-program.
On the last iteration, $B$ is false, and the operational invariant $\PrP$ still holds.
Hence I can assert $\PsP$ as a post-program for the \WHILE loop.

\begin{theorem}
\label{thm:while}
\whileLaw is sound, \ie if the hypothesis is true, then so is the conclusion.
\end{theorem}
\begin{proof}
  I establish, by induction on $i$, the following claim:
  \begin{quote}
  Assume the hypothesis \otp{\PrP'}{\Prg}{\PrP}, and that execution of the loop starts in a state
  $s_0 \in \pst{\PrP}$.
  Then if the loop executes for at least $i$ iterations,
  the state $s_i$ at the end of the $i$'th iteration is in $\pst{\PrP}$.
\end{quote}
Base case is for $i=0$: The state at the end of the 0'th iteration is just the start state $s_0$,
which is in  $\pst{\PrP}$ by assumption.\\
Induction step for $i > 0$: By the induction hypothesis, $s_{i-1} \in \pst{\PrP}$.
Now $s_{i-1}(B) = \true$, since otherwise the $i$'th iteration would not have been executed.
Hence $s_{i-1} \in \pst{\PrP'}$.   By the hypothesis \otp{\PrP'}{\Prg}{\PrP}, I have
$s_i \in \pst{\PrP}$.

Hence the claim is established.
I now show that \otp{\PrP}{\WHILEC{B} \Prg\ \ELIHW}{\PsP} holds.
Assume execution starts in an arbitrary $s \in \pst{\PrP}$ and that the loop terminates in some
state $t$. By the above claim, $t \in \pst{\PrP}$. Also, $t(B) = \false$ since otherwise the loop
cannot terminate in state $t$. Hence $t \in \pst{\PsP}$, and so
\otp{\PrP}{\WHILEC{B} \Prg\ \ELIHW}{\PsP} is established. Hence \whileLaw is sound.
\end{proof}

\paragraph{While Rule with Consequence.}
By applying Proposition~\ref{prop:ordering} and \postRl to \whileLaw,
I obtain \whileLawC, 
which states that the operational invariant can be a ``suffix'' of the actual
post-program of the loop body. This is often convenient, in practice.

\begin{theorem}
\label{thm:whileC}
\whileLawC is sound, \ie if the hypothesis is true, then so is the conclusion.
\end{theorem}
\begin{proof}
Assume the hypothesis \otp{\PrP'}{\Prg}{\gamma;\PrP}.
By   Proposition~\ref{prop:ordering}, $\gamma; \PrP \Pord \PrP$.
Hence by \postRl, \otp{\PrP'}{\Prg}{\PrP}.
Hence by \whileLaw, \otp{\PrP}{\WHILEC{B} \Prg \ELIHW}{\PsP}.
\end{proof}

\paragraph{If Rule.} 
Let \PrP be the pre-program.
Assume that execution of $\Prg1$ with pre-program $\PrP' \PeqB (\PrP, B)$
 leads to post-program \PsP,
and that execution of $\Prg2$ with pre-program $\PrP'' \PeqB (\PrP, \neg B)$
also leads to post-program \PsP. 
Then, execution of $\IF{B} \ \THEN  \Prg1 \ \ELSE \Prg2\ \FI$ with pre-program \PrP leads to
post-program \PsP.

\begin{theorem}
\label{thm:if}
\ifLaw is sound, \ie if both hypotheses are true, then so is the conclusion.
\end{theorem}
\bpr
Assume execution starts in an arbitrary $s \in \pst{\PrP}$.
Suppose that $B$ holds in $s$. Then, $s \in \pst{\PrP'}$. Also, the \IF\ branch will be taken,
so that $\Prg1$ is executed. Let $t$ be any resulting state. From the hypothesis
\otp{\PrP'}{\Prg1}{\PsP}, I have $t \in \pst{\PsP}$.
Now suppose that $B$ does not hold in $s$. Then $s \in \pst{\PrP''}$. Also, the \ELSE\ branch will be taken,
so that $\Prg2$ is executed. Let $t$ be any resulting state. From the hypothesis
\otp{\PrP''}{\Prg2}{\PsP}, I have $t \in \pst{\PsP}$.
By Definition~\ref{def:opTriple},
$\otp{\PrP}{\IF{B} \ \THEN  \Prg1 \ \ELSE \Prg2\ \FI}{\PsP}$ is valid.
\epr

\paragraph{One-way If Rule.} 
Assume that execution of $\Prg$ with pre-program $\PrP' \PeqB (\PrP, B)$
leads to post-program \PsP.
Assume that any post-state of $\PrP'' \PeqB (\PrP, \neg B)$ is also a post-state
of $\PsP$. Then execution of $\IF{B} \ \THEN  \Prg\ \FI$ with pre-program \PrP always leads to
post-program \PsP.

\begin{theorem}
\label{thm:if-oneway}
\ifLawO is sound, \ie if both hypotheses are true, then so is the conclusion.
\end{theorem}
\bpr
Assume execution starts in an arbitrary $s \in \pst{\PrP}$.
Suppose that $B$ holds in $s$. Then, $s \in \pst{\PrP'}$. Also, the \IF\ branch will be taken,
so that $\Prg$ is executed. Let $t$ be any resulting state. From the hypothesis
\otp{\PrP'}{\Prg}{\PsP}, I have $t \in \pst{\PsP}$.
Now suppose that $B$ does not hold in $s$. Then $s \in \pst{\PrP''}$.
Now $\PrP'' \Pord \PsP$, and so $s \in \pst{\PsP}$ by Def.~\ref{def:progOrd}.
Since the \IF\ branch is not taken, there is no
change of state, and so the resulting state $t$ is the same as $s$. Hence $t \in \pst{\PsP}$.
By Definition~\ref{def:opTriple}, \otp{\PrP}{\IF{B} \ \THEN  \Prg\ \FI}{\PsP}.
\epr

\section{Examples}

I now illustrate program verification with operational triples by means of several examples.
Throughout, I use an informally justified notion of program equivalence, based on well-known transformations such as eliminating the empty program \skp, and unwinding the last iteration of a \FOR loop.
Specifications programs are written in bold red italics, and regular programs are written in typewriter.
   \newcommand{\swp}{\leftrightarrow}
\newcommand{\lrk}{\mathit{findRank}}
\newcommand{\lrm}{\mathit{findMin}}
\newcommand{\vphi}{\varphi}

\newcommand{\rndi}{[0 : n-2]}
\newcommand{\rndip}{[1 : n-1]}

\subsection{Example: selection sort}
\label{app:sorting}




The input is an array $a$ and its size $n$, with $a$ indexed
from 0 to $n-1$. The assignment $i \asg [a:b]$ (with $a \le b$) nondeterministically
chooses a value between $a$ and $b$ inclusive, and assigns it to $i$.
The following holds by \emptPrg. I will develop this into the loop body for selection sort, which will be iterated for $i$ from 0 to $n-2$.
$\lrk(a,I)$ returns the index of the rank $I$ element in array $a$.
In case of duplicates, I take the element with lower index to also be of
lower rank. This breaks ties for duplicates, and also ensures that the
resulting sort is stable. We can specify a non-stable sort by
assigning the rank randomly amongst duplicate values (within the
appropriate range).
\begin{tabbing}
mmm\=mmm\=mmm\= \kill  
\lio{\spec{i \asg \rndi; }}
\lio{\spec{ \FORC{I = 0 \TO i}  J \asg \lrk(a,I); a[I] \swp a[J] \ROF}}
\lio{\reg{ \ttskp }}
\lio{\spec{i \asg \rndi; }}
\lion{\spec{ \FORC{I = 0 \TO i}  J \asg \lrk(a,I); a[I] \swp a[J] \ROF}}
\end{tabbing}
Now unwind the last iteration of the loop in the pre-program:
\begin{tabbing}
mmm\=mmm\=mmm\= \kill  
\lio{\spec{i \asg \rndi; }}
\lio{\spec{ \FORC{I = 0 \TO i}  J \asg \lrk(a,I); a[I] \swp a[J] \ROF}}
\lio{\spec{ J \asg \lrk(a,i); a[i] \swp a[J] }}
\lio{\reg{ \ttskp }}
\lio{\spec{i \asg \rndi; }}
\lion{\spec{ \FORC{I = 0 \TO i}  J \asg \lrk(a,I); a[I] \swp a[J] \ROF}}
\end{tabbing}
Now use \trading to move the third line of the pre-program into the
program.
Then remove the $skip$, as it is no longer needed.
\begin{tabbing}
  mmm\=mmm\=mmm\= \kill
\lio{\spec{i \asg \rndi; }}  
\lio{\spec{ \FORC{I = 0 \TO i-1}  J \asg \lrk(a,I); a[I] \swp a[J] \ROF}}
\lio{\reg{ J \asg \lrk(a,i); a[i] \swp a[J] }}
\lio{\spec{i \asg \rndi; }}
\lion{\spec{ \FORC{I = 0 \TO i}  J \asg \lrk(a,I); a[I] \swp a[J] \ROF}}
\end{tabbing}
Now I compute $\lrk(a,i)$ explicitly using an inner loop.
$\lrm(a,\vphi)$ computes the location of the lowest minimum element in
array $a$, excluding the consideration of elements whose indices are in the set $\vphi$.
\begin{tabbing}
mmm\=mmm\=mmm\= \kill
\lio{\spec{i \asg \rndi; }}  
\lio{\spec{ \FORC{I = 0 \TO i-1}  J \asg \lrk(a,I); a[I] \swp a[J] \ROF}}
\lio{\reg{ \vphi \asg \emptyset; \FORC{k = 0 \TO i-1}  \vphi \asg \vphi \un \lrm(a,\vphi) \ROF}} 
\lio{\reg{ J \asg \lrm(a,\vphi) } }
\lio{\reg{ a[i] \swp a[J] }}
\lio{\spec{i \asg \rndi; }}  
\lion{\spec{ \FORC{I = 0 \TO i}  J \asg \lrk(a,I); a[I] \swp a[J] \ROF}}
\end{tabbing}
Now use \trading to trade into the pre-program:
\begin{tabbing}
mmm\=mmm\=mmm\= \kill
\lio{\spec{i \asg \rndi; }}  
\lio{\spec{ \FORC{I = 0 \TO i-1}  J \asg \lrk(a,I); a[I] \swp a[J] \ROF}}
\lio{\spec{ \vphi \asg \emptyset; \FORC{k = 0 \TO i-1} \vphi \asg \vphi \un \lrm(a,\vphi)  \ROF}}
\lio{\reg{ J \asg \lrm(a,\vphi) } }
\lio{\reg{ a[i] \swp a[J] }}
\lio{\spec{i \asg \rndi; }}  
\lion{\spec{ \FORC{I = 0 \TO i}  J \asg \lrk(a,I); a[I] \swp a[J] \ROF}}
\end{tabbing}
Now construct a sequence of equivalences on the pre-program
\begin{tabbing}
mmm\=mmm\=mmm\= \kill  
\lio{\spec{ \FORC{I = 0 \TO i-1}  J \asg \lrk(a,I); a[I] \swp a[J] \ROF}}
\lio{\spec{ \vphi \asg \emptyset; \FORC{k = 0 \TO i-1} \vphi \asg \vphi \un \lrm(a,\vphi) \ROF}}
\lio{\Peqv}
\lio{\spec{ \vphi \asg \emptyset; }}
\lio{\spec{ \FORC{I = 0 \TO i-1}  J \asg \lrk(a,I); a[I] \swp a[J]; \vphi \asg \vphi \un \lrm(a,\vphi)  \ROF}}
\lio{\Peqv}
\lio{\spec{ \FORC{I = 0 \TO i-1}  J \asg \lrk(a,I); a[I] \swp a[J];  \vphi \asg \vphi \un I  \ROF}}
\lio{\Peqv}
\lio{\spec{ \FORC{I = 0 \TO i-1}  J \asg \lrk(a,I); a[I] \swp a[J] \ROF;}}
\lio{\spec{ \vphi \asg \{0,\ldots,i-1\}; }}
\end{tabbing}
I therefore have, using \eqSubRule
\begin{tabbing}
 mmm\=mmm\=mmm\= \kill
\lio{\spec{i \asg \rndi; }}  
\lio{\spec{ \FORC{I = 0 \TO i-1}  J \asg \lrk(a,I); a[I] \swp a[J] \ROF;}}
\lio{\spec{ \vphi \asg \{0,\ldots,i-1\}; }}
\lio{\reg{ J \asg \lrm(a,\vphi) } }
\lio{\reg{ a[i] \swp a[J] }}
\lio{\spec{i \asg \rndi; }}  
\lion{\spec{ \FORC{I = 0 \TO i}  J \asg \lrk(a,I); a[I] \swp a[J] \ROF}}
\end{tabbing}
Now use \trading to trade into the program
\begin{tabbing}
  mmm\=mmm\=mmm\= \kill
  \lio{\spec{i \asg \rndi; }}  
\lio{\spec{ \FORC{I = 0 \TO i-1}  J \asg \lrk(a,I); a[I] \swp a[J] \ROF;}}
\lio{\reg{ \vphi \asg \{0,\ldots,i-1\}; }}
\lio{\reg{ J \asg \lrm(a,\vphi) } }
\lio{\reg{ a[i] \swp a[J] }}
\lio{\spec{i \asg \rndi; }}  
\lion{\spec{ \FORC{I = 0 \TO i}  J \asg \lrk(a,I); a[I] \swp a[J] \ROF}}
\end{tabbing}
I now replace $\vphi \asg \{0,\ldots,i-1\};  J \asg \lrm(a,\vphi)$ by the equivalent
$J \asg \lrm(a, [0:i-1])$, so that $\lrm(a, [0:i-1])$ finds a minimum element
in array $a$ in the range $i$ to $n-1$, since the indices in $[0:i-1]$
are excluded. 
Hence, by \eqSubRule
\begin{tabbing}
  mmm\=mmm\=mmm\= \kill
  \lio{\spec{i \asg \rndi; }}  
\lio{\spec{ \FORC{I = 0 \TO i-1}  J \asg \lrk(a,I); a[I] \swp a[J] \ROF;}}
\lio{\reg{ J \asg \lrm(a, [0:i-1]); }}
\lio{\reg{ a[i] \swp a[J] }}
\lio{\spec{i \asg \rndi; }}  
\lion{\spec{ \FORC{I = 0 \TO i}  J \asg \lrk(a,I); a[I] \swp a[J] \ROF}}
\end{tabbing}
I now add the increment of $i$ at the end of the loop:
\begin{tabbing}
  mmm\=mmm\=mmm\= \kill
\lio{\spec{i \asg \rndi; }}  
\lio{\spec{ \FORC{I = 0 \TO i-1}  J \asg \lrk(a,I); a[I] \swp a[J] \ROF;}}
\lio{\reg{ J \asg \lrm(a, [0:i-1]); }}
\lio{\reg{ a[i] \swp a[J] }}
\lio{\spec{i \asg \rndi; }}  
\lio{\spec{ \FORC{I = 0 \TO i}  J \asg \lrk(a,I); a[I] \swp a[J] \ROF}}
\lio{\reg{ i \asg i+1 }}
\lio{\spec{i \asg \rndip; }}  
\lion{\spec{ \FORC{I = 0 \TO i-1}  J \asg \lrk(a,I); a[I] \swp a[J] \ROF;}}
\end{tabbing}
which gives the body of selection sort.
I finish up the example by including the outer loop
\begin{tabbing}
mmm\=mmm\=mmm\= \kill
\lio{\spec{\skp}}
\lio{\reg{i \asg 0}}
\lio{\spec{i \asg 0}}
\lio{\spec{ \FORC{I = 0 \TO i-1}  J \asg \lrk(a,I); a[I] \swp a[J] \ROF}}  

\lio{\reg{ \ttWHILEC{i \ne n-1} }}

\lit{\spec{i \asg \rndi; }}  
\lit{\spec{ \FORC{I = 0 \TO i-1}  J \asg \lrk(a,I); a[I] \swp a[J] \ROF}}
\lit{\reg{ J \asg \lrm(a, [0:i-1]); }}
\lit{\reg{ a[i] \swp a[J] }}
\lit{\spec{i \asg \rndi; }}  
\lit{\spec{ \FORC{I = 0 \TO i}  J \asg \lrk(a,I); a[I] \swp a[J] \ROF}}
\lit{\reg{ i \asg i+1 }}
\lit{\spec{i \asg \rndip; }}  
\lit{\spec{ \FORC{I = 0 \TO i-1}  J \asg \lrk(a,I); a[I] \swp a[J] \ROF;}}

\lio{\reg{ \ttELIHW }}

\lio{\spec{ i \asg n-1 } }
\lion{\spec{ \FORC{I = 0 \TO i}  J \asg \lrk(a,I); a[I] \swp a[J] \ROF}}
\end{tabbing}
By applying \postRl, I obtain
\begin{tabbing}
mmm\=mmm\=mmm\= \kill
\lio{\spec{\skp}}
\lio{\reg{i \asg 0}}
\lio{\spec{i \asg 0}}
\lio{\spec{ \FORC{I = 0 \TO i-1}  J \asg \lrk(a,I); a[I] \swp a[J] \ROF}}  

\lio{\reg{ \ttWHILEC{i \ne n-1} }}

\lit{\spec{i \asg \rndi; }}  
\lit{\spec{ \FORC{I = 0 \TO i-1}  J \asg \lrk(a,I); a[I] \swp a[J] \ROF}}
\lit{\reg{ J \asg \lrm(a, [0:i-1]); }}
\lit{\reg{ a[i] \swp a[J] }}
\lit{\spec{i \asg \rndi; }}  
\lit{\spec{ \FORC{I = 0 \TO i}  J \asg \lrk(a,I); a[I] \swp a[J] \ROF}}
\lit{\reg{ i \asg i+1 }}
\lit{\spec{i \asg \rndip; }}  
\lit{\spec{ \FORC{I = 0 \TO i-1}  J \asg \lrk(a,I); a[I] \swp a[J] \ROF;}}

\lio{\reg{ \ttELIHW }}

\lion{\spec{ \FORC{I = 0 \TO n-1}  J \asg \lrk(a,I); a[I] \swp a[J] \ROF}}
\end{tabbing}
The post-program states that the element at index $I$ has rank $I$,
that is, array $a$ is sorted.
I omit the development of $\lrm$.
   \newcounter{codectr}

\newcommand{\rnds}{?}
\newcommand{\rndsp}{?'}

\subsection{Example: Dijkstra's shortest paths algorithm}

The input consists of the following:
\be
\item a directed graph $G = (V,E)$ with node set
$V$ and edge set $E$, together with a weight function $w:E \to
\Reals^{\ge 0}$, since Dijkstra's algorithm  requires that edge weights are non-negative.
\item a distinguished vertex $s$: the source.
\item for each node $t$, a real number $t.d$, which records the current estimate of 
   the shortest path distance from the source $s$ to $t$
\ee
I use $\delta[t]$ to denote the shortest path distance from the source
$s$ to node $t$, and $u \ar v$ to mean $(u,v) \in E$, \ie there is an
edge in $G$ from $u$ to $v$.

Dijkstra's algorithm is incremental, each iteration of the main loop
computes the shortest path distance of some node $m$.  Vertices are
colored black (shortest path distance from the source $s$ is known),
grey (have an incoming edge form a black node), and white (neither
black nor grey). The algorithm maintains the sets $B$ and $G$ of black
and grey nodes.

\reset


I start with an instance of \emptPrg.
The assignment $B \asg \rnds$ nondeterministically sets $B$ to any subset of $V$ which
contains the source $s$. It can be implemented using the nondeterministic choice operator $\ch$
(Section~\ref{sec:prelim}).
The pre- and post-programs
simply assign the correct shortest path distances to the $t.d$
variable for each node $t$ in $B$.
\begin{tabbing}
  mmm\=mmm\=mmm\= \kill
\lio{\spec{B \asg \rnds}}  
\lio{\spec{\FORC{t \in B} t.d \asg \delta[t] \ROF; \FORC{t \nin B} t.d \asg +\infty \ROF}}
\lio{\reg{\ttskp}}
\lio{\spec{B \asg \rnds}}  
\lion{\spec{\FORC{t \in B} t.d \asg \delta[t] \ROF; \FORC{t \nin B} t.d \asg +\infty \ROF}}
\end{tabbing}
Now append to both the  pre- and post-programs a statement which relaxes all of the grey nodes.
This is still an instance of \emptPrg.
\begin{tabbing}
  mmm\=mmm\=mmm\= \kill
  \lio{\spec{B \asg \rnds}}  
\lio{\spec{\FORC{t \in B} t.d \asg \delta[t] \ROF; \FORC{t \nin B} t.d \asg +\infty \ROF}}
\lio{\spec{\FORC{g \in G} \FORC{b \in B \land b \ar g} relax(b,g) \ROF \ROF}\setcounter{codectr}{\arabic{lctr}}}
\lio{\reg{\ttskp}}
\lio{\spec{B \asg \rnds}}  
\lio{\spec{\FORC{t \in B} t.d \asg \delta[t] \ROF; \FORC{t \nin B} t.d \asg +\infty \ROF}}
\lion{\spec{\FORC{g \in G} \FORC{b \in B \land b \ar g} relax(b,g) \ROF \ROF}}
\end{tabbing}

\remove{
Now apply \trading to move line~\arabic{codectr} from the pre-program to the program.
Also delete the \skp\ as it is now superfluous.
\begin{tabbing}
mmm\=mmm\=mmm\= \kill
\lio{\spec{B \asg \rnds}}  
\lio{\spec{\FORC{t \in B} t.d \asg \delta[t] \ROF; \FORC{t \nin B} t.d \asg +\infty \ROF}}

\lio{\reg{\ttFORC{g \in G} \ttFORC{b \in B \land b \ar g} relax(b,g) \ttROF \ttROF}}

\lio{\spec{B \asg \rnds}}  
\lio{\spec{\FORC{t \in B} t.d \asg \delta[t] \ROF; \FORC{t \nin B} t.d \asg +\infty \ROF}}
\lion{\spec{\FORC{g \in G} \FORC{b \in B \land b \ar g} relax(b,g) \ROF \ROF}}
\end{tabbing}
}

\remove{
Apply \appRule:
\begin{tabbing}
mmm\=mmm\=mmm\= \kill
\lio{\spec{B \asg \rnds}}    
\lio{\spec{\FORC{t \in B} t.d \asg \delta[t] \ROF; \FORC{t \nin B} t.d \asg +\infty \ROF}}
\lio{\spec{\FORC{t \in B} t.d \asg \delta[t] \ROF; \FORC{t \nin B} t.d \asg +\infty \ROF}}


\lio{\reg{m \asg \{ g \sp g.d = (\MIN h \in G: h.d) \}}}

\lio{\spec{B \asg \rnds}}  
\lio{\spec{\FORC{t \in B} t.d \asg \delta[t] \ROF; \FORC{t \nin B} t.d \asg +\infty \ROF}}
\lio{\spec{\FORC{g \in G} \FORC{b \in B \land b \ar g} relax(b,g) \ROF \ROF}}
\lio{\spec{m \asg \{ g \sp g.d = (\MIN h \in G: h.d) \}}}
\end{tabbing}
}

\noindent
Apply \appRule to add
``$m \asg \{ g \sp g.d = (\MIN h \in G: h.d)$;  
$m.d \asg (\MIN \pi \in paths(s, m) : |\pi|)$''
to both program and post-program, and also remove the \ttskp, as it is no longer needed.
$|\pi|$ is the total cost of path $\pi$, and  $paths(s, m)$ is the set of all simple paths from $s$
to $m$.
\begin{tabbing}
mmm\=mmm\=mmm\= \kill  
\lio{\spec{B \asg \rnds}}  
\lio{\spec{\FORC{t \in B} t.d \asg \delta[t] \ROF; \FORC{t \nin B} t.d \asg +\infty \ROF}}
\lio{\spec{\FORC{g \in G} \FORC{b \in B \land b \ar g} relax(b,g) \ROF \ROF}}

\lio{\reg{m \asg \{ g \sp g.d = (\MIN h \in G: h.d) \}}}
\lio{\reg{m.d \asg (\MIN \pi \in paths(s, m) : |\pi|)}}

\lio{\spec{B \asg \rnds}}  
\lio{\spec{\FORC{t \in B} t.d \asg \delta[t] \ROF; \FORC{t \nin B} t.d \asg +\infty \ROF}}
\lio{\spec{\FORC{g \in G} \FORC{b \in B \land b \ar g} relax(b,g) \ROF \ROF}}
\lio{\spec{m \asg \{ g \sp g.d = (\MIN h \in G: h.d) \}}}
\lio{\spec{m.d \asg (\MIN \pi \in paths(s, m) : |\pi|)}}
\end{tabbing}
%


Now construct a sequence of equivalences as follows:
\begin{tabbing}
mmm\=mmm\=mmm\= \kill  
\lio{\spec{\FORC{g \in G} \FORC{b \in B \land b \ar g} relax(b,g) \ROF \ROF}}
\lioc{\Peqv}{definition of relax(b,g)}
\lio{\spec{\FORC{g \in G} \FORC{b \in B \land b \ar g} g.d \asg g.d \min b.d + w(b,g) ) \ROF \ROF}}
\lioc{\Peqv}{min is commutative and associative}
\lio{\spec{\FORC{g \in G}   g.d \asg g.d \min\ (\MIN b \in B \land b \ar g: b.d + w(b,g))  \ROF \ROF}}
\lioc{\Peqv}{use the operational invariant}
\lio{\spec{\FORC{g \in G}   g.d \asg g.d \min\ (\MIN b \in B \land b \ar g: \delta[b] + w(b,g))  \ROF}}
\lioc{\Peqv}{$s \la{B} g$ denotes all paths from $s$ to $g$ containing only black nodes except $g$}
\lio{\spec{\FORC{g \in G}   g.d \asg g.d \min\ (\MIN \pi \in s \la{B} g: |\pi|)  \ROF}}
\lioc{\Peqv}{$s \la{} g$ denotes all paths from $s$ to $g$}
\lioc{}{let $M = \{g \in G : g.d = (\MIN h \in G: h.d)\}$, use algorithm assumption}
\lio{\spec{\FORC{g \in G}   g.d \asg g.d \min\ (\MIN \pi \in s \la{B} g: |\pi|)  \ROF;}}
\lio{\spec{\FORC{g \in M}   g.d \asg g.d \min\ (\MIN \pi \in s \la{B} g: |\pi|) \min\ (\MIN \pi \in s \la{} g: |\pi|) \ROF}}
\lioc{\Peqv}{definition of $\delta$}
\lio{\spec{\FORC{g \in G}   g.d \asg g.d \min\ (\MIN \pi \in s \la{B} g: |\pi|)  \ROF;}}
\lio{\spec{\FORC{g \in M}   g.d \asg \delta[g]}}

\end{tabbing}
The above equivalence allows me to conclude the following:
\begin{tabbing}
mmm\=mmm\=mmm\= \kill  
\lio{\spec{B \asg \rnds}}  
\lio{\spec{\FORC{t \in B} t.d \asg \delta[t] \ROF; \FORC{t \nin B} t.d \asg +\infty \ROF}}
\lio{\spec{\FORC{g \in G} \FORC{b \in B \land b \ar g} relax(b,g) \ROF \ROF}}
\lio{\reg{m \asg \{ g \sp g.d = (\MIN h \in G: h.d) \}}}
\lion{\reg{m.d \asg (\MIN \pi \in paths(s, m) : |\pi|)}}
\end{tabbing}
is equivalent to
\begin{tabbing}
mmm\=mmm\=mmm\= \kill  
\lio{\spec{B \asg \rnds}}  
\lio{\spec{\FORC{t \in B} t.d \asg \delta[t] \ROF; \FORC{t \nin B} t.d \asg +\infty \ROF}}
\lio{\spec{\FORC{g \in G} \FORC{b \in B \land b \ar g} relax(b,g) \ROF \ROF}}
\lion{\reg{m \asg \{ g \sp g.d = (\MIN h \in G: h.d) \}}}
\end{tabbing}
%
%
I also observe that \spec{m.d \asg (\MIN \pi \in paths(s, m) : |\pi|)}
is equivalent to \spec{m.d \asg \delta[m]} 
by definition of $\delta[m]$, the shortest path distance from source $s$ to node $m$.
Hence, by \eqSubRule I can now write
\begin{tabbing}
mmm\=mmm\=mmm\= \kill
\lio{\spec{B \asg \rnds}}  
\lio{\spec{\FORC{t \in B} t.d \asg \delta[t] \ROF; \FORC{t \nin B} t.d \asg +\infty \ROF}}
\lio{\spec{\FORC{g \in G} \FORC{b \in B \land b \ar g} relax(b,g) \ROF \ROF}}

\lio{\reg{m \asg \{ g \sp g.d = (\MIN h \in G: h.d) \}}}

\lio{\spec{B \asg \rnds}}  
\lio{\spec{\FORC{t \in B} t.d \asg \delta[t] \ROF; \FORC{t \nin B} t.d \asg +\infty \ROF}}
\lio{\spec{\FORC{g \in G} \FORC{b \in B \land b \ar g} relax(b,g) \ROF \ROF}}
\lio{\spec{m \asg \{ g \sp g.d = (\MIN h \in G: h.d) \}}}
\lion{\spec{m.d \asg \delta[m]}}
\end{tabbing}
Now color $m$ black. This causes, in general, some nodes to turn grey, and so requires the addition of the line
$\reg{\FORC{g \in G \land m \ar g} relax(m,g)}$ to preserve the truth of the operational triple.
I retain the intermediate specification program between the program statement that selects $m$, and
the statements which turn $m$ black and then relax all of $m$'s grey neighbors.
The assignment $B \asg \rndsp$ nondeterministically sets $B$ to any subset of $V$
which contains both $s$ and $m$.
\begin{tabbing}
mmm\=mmm\=mmm\= \kill  
\lio{\spec{B \asg \rnds}}  
\lio{\spec{\FORC{t \in B} t.d \asg \delta[t] \ROF; \FORC{t \nin B} t.d \asg +\infty \ROF}}
\lio{\spec{\FORC{g \in G} \FORC{b \in B \land b \ar g} relax(b,g) \ROF \ROF}}

\lio{\reg{m \asg \{ g \sp g.d = (\MIN h \in G: h.d) \}}}

\lio{\spec{B \asg \rnds}}  
\lio{\spec{\FORC{t \in B} t.d \asg \delta[t] \ROF; \FORC{t \nin B} t.d \asg +\infty \ROF}}
\lio{\spec{\FORC{g \in G} \FORC{b \in B \land b \ar g} relax(b,g) \ROF \ROF}}
\lio{\spec{m \asg \{ g \sp g.d = (\MIN h \in G: h.d) \}}}
\lio{\spec{m.d \asg \delta[m]}}

\lio{\reg{B \asg B \un \set{m};}}
\lio{\reg{\ttFORC{g \in G \land m \ar g} relax(m,g)}} 

\lio{\spec{B \asg \rndsp}}  
\lio{\spec{\FORC{t \in B} t.d \asg \delta[t] \ROF; \FORC{t \nin B} t.d \asg +\infty \ROF}}
\lion{\spec{\FORC{g \in G} \FORC{b \in B \land b \ar g} relax(b,g) \ROF \ROF}}
\end{tabbing}
This concludes the development of the loop body. The pre-program consists of making the looping
condition true followed by the specification program
\begin{tabbing}
mmm\=mmm\=mmm\= \kill  
\lio{\spec{\FORC{t \in B} t.d \asg \delta[t] \ROF; \FORC{t \nin B} t.d \asg +\infty \ROF}}
\lion{\spec{\FORC{g \in G} \FORC{b \in B \land b \ar g} relax(b,g) \ROF \ROF}}
\end{tabbing}
and the post-program has this specification program as a suffix. Hence this specification program
plays the role of an ``invariant'', and we have, by \whileLawC, the complete annotated program:
\begin{tabbing}
mmm\=mmm\=mmm\= \kill
\lio{\spec{skip}}
\lio{\reg{B \asg \set{s}; s.d \asg 0; \FORC{t \nin B} t.d \asg +\infty \ROF} }
\lio{\spec{\FORC{t \in B} t.d \asg \delta[t] \ROF; \FORC{t \nin B} t.d \asg +\infty \ROF}}

\lio{\reg{\ttWHILEC{B \ne V}}}

\lit{\spec{B \asg \rnds}}  
\lit{\spec{\FORC{t \in B} t.d \asg \delta[t] \ROF; \FORC{t \nin B} t.d \asg +\infty \ROF}}
\lit{\spec{\FORC{g \in G} \FORC{b \in B \land b \ar g} relax(b,g) \ROF \ROF}}

\lit{\reg{m \asg \{ g \sp g.d = (\MIN h \in G: h.d) \}}}

\lit{\spec{B \asg \rnds}}  
\lit{\spec{\FORC{t \in B} t.d \asg \delta[t] \ROF; \FORC{t \nin B} t.d \asg +\infty \ROF}}
\lit{\spec{\FORC{g \in G} \FORC{b \in B \land b \ar g} relax(b,g) \ROF \ROF}}
\lit{\spec{m \asg \{ g \sp g.d = (\MIN h \in G: h.d) \}}}
\lit{\spec{m.d \asg \delta[m]}}

\lit{\reg{B \asg B \un \set{m}}}
\lit{\reg{\ttFORC{g \in G \land m \ar g} relax(m,g)}} 

\lit{\spec{B \asg \rndsp}}  
\lit{\spec{\FORC{t \in B} t.d \asg \delta[t] \ROF; \FORC{t \nin B} t.d \asg +\infty \ROF}}
\lit{\spec{\FORC{g \in G} \FORC{b \in B \land b \ar g} relax(b,g) \ROF \ROF}}

\lio{\reg{\ttELIHW}}
\lio{\spec{B \asg V}}
\lio{\spec{\FORC{t \in B} t.d \asg \delta[t] \ROF; \FORC{t \nin B} t.d \asg +\infty \ROF}}
\lion{\spec{\FORC{g \in G} \FORC{b \in B \land b \ar g} relax(b,g) \ROF \ROF}}
\end{tabbing}
By applying \postRl, I obtain
\begin{tabbing}
mmm\=mmm\=mmm\= \kill
\lio{\spec{skip}}
\lio{\reg{B \asg \set{s}; s.d \asg 0; \FORC{t \nin B} t.d \asg +\infty \ROF} }
\lio{\spec{\FORC{t \in B} t.d \asg \delta[t] \ROF; \FORC{t \nin B} t.d \asg +\infty \ROF}}

\lio{\reg{\ttWHILEC{B \ne V}}}

\lit{\spec{B \asg \rnds}}  
\lit{\spec{\FORC{t \in B} t.d \asg \delta[t] \ROF; \FORC{t \nin B} t.d \asg +\infty \ROF}}
\lit{\spec{\FORC{g \in G} \FORC{b \in B \land b \ar g} relax(b,g) \ROF \ROF}}

\lit{\reg{m \asg \{ g \sp g.d = (\MIN h \in G: h.d) \}}}

\lit{\spec{B \asg \rnds}}  
\lit{\spec{\FORC{t \in B} t.d \asg \delta[t] \ROF; \FORC{t \nin B} t.d \asg +\infty \ROF}}
\lit{\spec{\FORC{g \in G} \FORC{b \in B \land b \ar g} relax(b,g) \ROF \ROF}}
\lit{\spec{m \asg \{ g \sp g.d = (\MIN h \in G: h.d) \}}}
\lit{\spec{m.d \asg \delta[m]}}

\lit{\reg{B \asg B \un \set{m}}}
\lit{\reg{\ttFORC{g \in G \land m \ar g} relax(m,g)}} 

\lit{\spec{B \asg \rndsp}}  
\lit{\spec{\FORC{t \in B} t.d \asg \delta[t] \ROF; \FORC{t \nin B} t.d \asg +\infty \ROF}}
\lit{\spec{\FORC{g \in G} \FORC{b \in B \land b \ar g} relax(b,g) \ROF \ROF}}

\lio{\reg{\ttELIHW}}
\lion{\spec{\FORC{t \in V} t.d \asg \delta[t] \ROF}}
\end{tabbing}
The post-program sets the $t.d$ variable for every node $t$ to the
correct shortest path value.

   \newcommand{\rndl}{[0:\l-3]}
\newcommand{\rndlp}{[1:\l-2]}

\subsection{Example: in-place list reversal}

I now illustrate the use of operational annotations to derive a correct
algorithm for the in-place reversal of a linked list.
The input is a size $\l > 0$ array $n$ of objects of type $Node$, indexed
from 0 to $\l-1$.
$Node$ is declared as follows:
$\spec{class\ Node \{ Node\ p;\ \mbox{other fields} \ldots\}}$.
Element $i$ is referred to as $n_i$ instead of $n[i]$, and
contains a pointer $n_i.p$, and possibly other (omitted) fields.
The use of this array is purely for specification purposes, so that I
can go through the nodes and construct the linked list, which is then
reversed. An array also ensures that there is no aliasing: all elements
are distinct, by construction.
Array $n$ is created by executing {\spec{Node[\,]\ n \asg new\ Node[\l]}}.
The pre-program and post-program both start with code to declare
$Node$, followed by the above line to create array $n$.
I omit this code as including it would be repetitive
and would add clutter.

I start by applying \emptPrg. The pre (and post) programs do 3 things:
(1) construct the linked list, (2) reverse part of the list, up to
position $i+1$, and (3) maintain 3 pointers, into positions $i+1$,
$i+2$, and $i+3$. 
\reset
\begin{tabbing}
mmm\=mmm\=mmm\= \kill
\lio{\spec{ i \asg \rndl }}
\lio{\spec{ \FORC{j = 0 \TO \l-1}  n_j.p \asg n_{j+1} \ROF }}
\lio{\spec{ \FORC{j = i+1 \DOWNTO 1}  n_j.p \asg n_{j-1}; \ROF }}
\lio{\spec{ r \asg n_{i+1}; s \asg n_{i+2}; t \asg n_{i+3} }}

\lio{\reg{ \ttskp }}

\lio{\spec{ i \asg \rndl }}
\lio{\spec{ \FORC{j = 0 \TO \l-1}  n_j.p \asg n_{j+1} \ROF }}
\lio{\spec{ \FORC{j = i+1 \DOWNTO 1}  n_j.p \asg n_{j-1}; \ROF }}
\lion{\spec{ r \asg n_{i+1}; s \asg n_{i+2}; t \asg n_{i+3} }}
\end{tabbing}
Now unwind the last iteration of the second \FOR loop of the
pre-program.
So, by \eqSubRule
\begin{tabbing}
mmm\=mmm\=mmm\= \kill
\lio{\spec{ i \asg \rndl }}
\lio{\spec{ \FORC{j = 0 \TO \l-1}  n_j.p \asg n_{j+1} \ROF }}
\lio{\spec{ \FORC{j = i \DOWNTO 1}  n_j.p \asg n_{j-1}; \ROF }}
\lio{\spec{ n_{i+1}.p \asg n_{i}; }}
\lio{\spec{ r \asg n_{i+1}; s \asg n_{i+2}; t \asg n_{i+3} }}

\lio{\reg{ \ttskp }}

\lio{\spec{ i \asg \rndl }}
\lio{\spec{ \FORC{j = 0 \TO \l-1}  n_j.p \asg n_{j+1} \ROF }}
\lio{\spec{ \FORC{j = i+1 \DOWNTO 1}  n_j.p \asg n_{j-1}; \ROF }}
\lion{\spec{ r \asg n_{i+1}; s \asg n_{i+2}; t \asg n_{i+3} }}
\end{tabbing}
Now introduce the line
$r \asg n_{i}; s \asg n_{i+1}; t \asg n_{i+2}$ into the pre-program. Since $r, s, t$ are subsequently overwritten, and not referenced in the interim,
this preserves equivalence of the pro-program with its previous version.
So, by \eqSubRule
\begin{tabbing}
mmm\=mmm\=mmm\= \kill  
\lio{\spec{ i \asg \rndl }}
\lio{\spec{ \FORC{j = 0 \TO \l-1}  n_j.p \asg n_{j+1} \ROF }}
\lio{\spec{ \FORC{j = i \DOWNTO 1}  n_j.p \asg n_{j-1}; \ROF }}
\lio{\spec{ r \asg n_{i}; s \asg n_{i+1}; t \asg n_{i+2} }}
\lio{\spec{ n_{i+1}.p \asg n_{i}; }}
\lio{\spec{ r \asg n_{i+1}; s \asg n_{i+2}; t \asg n_{i+3} }}

\lio{\reg{ \ttskp }}

\lio{\spec{ i \asg \rndl }}
\lio{\spec{ \FORC{j = 0 \TO \l-1}  n_j.p \asg n_{j+1} \ROF }}
\lio{\spec{ \FORC{j = i+1 \DOWNTO 1}  n_j.p \asg n_{j-1}; \ROF }}
\lion{\spec{ r \asg n_{i+1}; s \asg n_{i+2}; t \asg n_{i+3} }}
\end{tabbing}
Now apply \trading to trade into the program, and remove the $\skp$ since it is no longer needed.
\begin{tabbing}
mmm\=mmm\=mmm\= \kill  
\lio{\spec{ i \asg \rndl }}
\lio{\spec{ \FORC{j = 0 \TO \l-1}  n_j.p \asg n_{j+1} \ROF }}
\lio{\spec{ \FORC{j = i \DOWNTO 1}  n_j.p \asg n_{j-1}; \ROF }}
\lio{\spec{ r \asg n_{i}; s \asg n_{i+1}; t \asg n_{i+2} }}

\lio{\reg{ n_{i+1}.p \asg n_{i}; }}
\lio{\reg{ r \asg n_{i+1}; s \asg n_{i+2}; t \asg n_{i+3} }}

\lio{\spec{ i \asg \rndl }}
\lio{\spec{ \FORC{j = 0 \TO \l-1}  n_j.p \asg n_{j+1} \ROF }}
\lio{\spec{ \FORC{j = i+1 \DOWNTO 1}  n_j.p \asg n_{j-1}; \ROF }}
\lion{\spec{ r \asg n_{i+1}; s \asg n_{i+2}; t \asg n_{i+3} }}
\end{tabbing}
Since $\spec{r \asg n_{i}; s \asg n_{i+1}}$ immediately precedes $\reg{n_{i+1}.p \asg n_{i}}$, I can
replace $\reg{n_{i+1}.p \asg n_{i}}$ by $\reg{s.p \asg r}$ while retaining equivalence.
So, by \eqSubRule
\begin{tabbing}
mmm\=mmm\=mmm\= \kill  
\lio{\spec{ i \asg \rndl }}
\lio{\spec{ \FORC{j = 0 \TO \l-1}  n_j.p \asg n_{j+1} \ROF }}
\lio{\spec{ \FORC{j = i \DOWNTO 1}  n_j.p \asg n_{j-1}; \ROF }}
\lio{\spec{ r \asg n_{i}; s \asg n_{i+1}; t \asg n_{i+2} }}

\lio{\reg{ s.p \asg r; }}
\lio{\reg{ r \asg n_{i+1}; s \asg n_{i+2}; t \asg n_{i+3} }}

\lio{\spec{ i \asg \rndl }}
\lio{\spec{ \FORC{j = 0 \TO \l-1}  n_j.p \asg n_{j+1} \ROF }}
\lio{\spec{ \FORC{j = i+1 \DOWNTO 1}  n_j.p \asg n_{j-1}; \ROF }}
\lion{\spec{ r \asg n_{i+1}; s \asg n_{i+2}; t \asg n_{i+3} }}
\end{tabbing}
Since  $\spec{s \asg n_{i+1}}$ precedes $\reg{ r \asg n_{i+1}}$ and $s$ is not modified in the interim, I
can replace $\reg{ r \asg n_{i+1}}$ by  $\reg{ r \asg s}$  while retaining equivalence.
So, by \eqSubRule
\begin{tabbing}
mmm\=mmm\=mmm\= \kill
\lio{\spec{ i \asg \rndl }}
\lio{\spec{ \FORC{j = 0 \TO \l-1}  n_j.p \asg n_{j+1} \ROF }}
\lio{\spec{ \FORC{j = i \DOWNTO 1}  n_j.p \asg n_{j-1}; \ROF }}
\lio{\spec{ r \asg n_{i}; s \asg n_{i+1}; t \asg n_{i+2} }}

\lio{\reg{ s.p \asg r; }}
\lio{\reg{ r \asg s; s \asg n_{i+2}; t \asg n_{i+3} }}

\lio{\spec{ i \asg \rndl }}
\lio{\spec{ \FORC{j = 0 \TO \l-1}  n_j.p \asg n_{j+1} \ROF }}
\lio{\spec{ \FORC{j = i+1 \DOWNTO 1}  n_j.p \asg n_{j-1}; \ROF }}
\lion{\spec{ r \asg n_{i+1}; s \asg n_{i+2}; t \asg n_{i+3} }}
\end{tabbing}
In a similar manner, I can replace $\reg{s \asg n_{i+2}}$ by  $\reg{s \asg t}$.
So, by \eqSubRule
\begin{tabbing}
mmm\=mmm\=mmm\= \kill  
\lio{\spec{ i \asg \rndl }}
\lio{\spec{ \FORC{j = 0 \TO \l-1}  n_j.p \asg n_{i+1} \ROF }}
\lio{\spec{ \FORC{j = i \DOWNTO 1}  n_j.p \asg n_{j-1}; \ROF }}
\lio{\spec{ r \asg n_{i}; s \asg n_{i+1}; t \asg n_{i+2} }}

\lio{\reg{ s.p \asg r; }}
\lio{\reg{ r \asg s; s \asg t; t \asg n_{i+3} }}

\lio{\spec{ i \asg \rndl }}
\lio{\spec{ \FORC{j = 0 \TO \l-1}  n_j.p \asg n_{i+1} \ROF }}
\lio{\spec{ \FORC{j = i+1 \DOWNTO 1}  n_j.p \asg n_{j-1}; \ROF }}
\lion{\spec{ r \asg n_{i+1}; s \asg n_{i+2}; t \asg n_{i+3} }}
\end{tabbing}
From $\spec{ \FORC{j = 0 \TO \l-1}  n_j.p \asg n_{i+1} \ROF }$, I have $\spec{n_{i+2}.p \asg n_{i+3}}$, and I observe
that $n_{i+2}.p$ is not subsequently modified. 
Also I have $\spec{ t \asg n_{i+2} }$ occurring before the program, and $t$ is not modified until
$\reg{t \asg n_{i+3}}$, we can replace $\reg{t \asg n_{i+3}}$ by $\reg{t \asg t.p}$.
So, by \eqSubRule
\begin{tabbing}
mmm\=mmm\=mmm\= \kill
\lio{\spec{ i \asg \rndl }}
\lio{\spec{ \FORC{j = 0 \TO \l-1}  n_j.p \asg n_{j+1} \ROF }}
\lio{\spec{ \FORC{j = i \DOWNTO 1}  n_j.p \asg n_{j-1}; \ROF }}
\lio{\spec{ r \asg n_{i}; s \asg n_{i+1}; t \asg n_{i+2} }}

\lio{\reg{ s.p \asg r; }}
\lio{\reg{ r \asg s; s \asg t; t \asg t.p }}

\lio{\spec{ i \asg \rndl }}
\lio{\spec{ \FORC{j = 0 \TO \l-1}  n_j.p \asg n_{j+1} \ROF }}
\lio{\spec{ \FORC{j = i+1 \DOWNTO 1}  n_j.p \asg n_{j-1}; \ROF }}
\lion{\spec{ r \asg n_{i+1}; s \asg n_{i+2}; t \asg n_{i+3} }}
\end{tabbing}
Now apply \whileLaw to obtain the complete program, while also
incrementing the loop counter $i$ at the end of the loop body.
\begin{tabbing}
mmm\=mmm\=mmm\= \kill  
\lio{\spec{ \FORC{j = 0 \TO \l-1}  n_j.p \asg n_{j+1} \ROF }}

\lio{\reg{ r \asg n_0; s \asg n_{1}; t \asg n_{2} }}

\lio{\spec{ i \asg 0 }}
\lio{\spec{ \FORC{j = 0 \TO \l-1}  n_j.p \asg n_{i+1} \ROF }}
\lio{\spec{ \FORC{j = i \DOWNTO 1}  n_j.p \asg n_{j-1}; \ROF }}
\lio{\spec{ r \asg n_{i}; s \asg n_{i+1}; t \asg n_{i+2} }}

\lio{\reg{\ttWHILEC{i \ne \l-2}}}

\lit{\spec{ i \asg \rndl }}
\lit{\spec{ \FORC{j = 0 \TO \l-1}  n_j.p \asg n_{j+1} \ROF }}
\lit{\spec{ \FORC{j = i \DOWNTO 1}  n_j.p \asg n_{j-1}; \ROF }}
\lit{\spec{ r \asg n_{i}; s \asg n_{i+1}; t \asg n_{i+2} }}

\lit{\reg{ s.p \asg r; }}
\lit{\reg{ r \asg s; s \asg t; t \asg t.p }}

\lit{\spec{ i \asg \rndl }}
\lit{\spec{ \FORC{j = 0 \TO \l-1}  n_j.p \asg n_{j+1} \ROF }}
\lit{\spec{ \FORC{j = i+1 \DOWNTO 1}  n_j.p \asg n_{j-1}; \ROF }}
\lit{\spec{ r \asg n_{i+1}; s \asg n_{i+2}; t \asg n_{i+3} }}

\lit{\reg{ i \asg i+1; }}

\lit{\spec{ i \asg \rndlp }}
\lit{\spec{ \FORC{j = 0 \TO \l-1}  n_j.p \asg n_{j+1} \ROF }}
\lit{\spec{ \FORC{j = i \DOWNTO 1}  n_j.p \asg n_{j-1}; \ROF }}
\lit{\spec{ r \asg n_{i}; s \asg n_{i+1}; t \asg n_{i+2} }}

\lio{\reg{\ttELIHW}}

\lio{\spec{i \asg \l-2}}
\lio{\spec{ r \asg n_{i}; s \asg n_{i+1}; t \asg n_{i+2} }}
\lio{\spec{ \FORC{j = 0 \TO \l-1}  n_j.p \asg n_{j+1} \ROF }}
\lion{\spec{ \FORC{j = i \DOWNTO 1}  n_j.p \asg n_{j-1}; \ROF }}
\end{tabbing}

The loop must terminate at $i=\l-2$ to avoid dereferencing $\nil$.
The post-program of the loop then gives 
$\spec{ \FORC{j = \l-2 \DOWNTO 1}  n_j.p \asg n_{j-1}; \ROF }$,
which means that the last node's pointer is not set to the previous node, since
the list ends at index $\l-1$.
Hence we require a final assignment that is equivalent to
$n_{\l-1}.p \asg n_{\l-2}$.

The post-program of the loop gives 
\spec{ i \asg \l-2; r \asg n_{i}; s \asg n_{i+1}; t \asg n_{i+2} }, which yields
\spec{ i \asg \l-2; r \asg n_{\l-2}; s \asg n_{\l-1}; t \asg n_{\l} }.
Hence the last assignment can be rendered as $s.p \asg r$.
Also, the loop termination condition can be rewritten as $t \ne \nil$,
since $t$ becomes \nil when it is assigned $n_{\l}$, which happens exactly when $i$ becomes $\l-2$.
Hence, using \eqSubRule and \postRl, I obtain

\begin{tabbing}
mmm\=mmm\=mmm\= \kill  
\lio{\spec{ \FORC{j = 0 \TO \l-1}  n_j.p \asg n_{j+1} \ROF }}

\lio{\reg{ r \asg n_0; s \asg n_{1}; t \asg n_{2} }}

\lio{\spec{ i \asg 0 }}
\lio{\spec{ \FORC{j = 0 \TO \l-1}  n_j.p \asg n_{i+1} \ROF }}
\lio{\spec{ \FORC{j = i \DOWNTO 1}  n_j.p \asg n_{j-1}; \ROF }}
\lio{\spec{ r \asg n_{i}; s \asg n_{i+1}; t \asg n_{i+2} }}

\lio{\reg{\ttWHILEC{t \ne \nil}}}

\lit{\spec{ i \asg \rndl }}
\lit{\spec{ \FORC{j = 0 \TO \l-1}  n_j.p \asg n_{j+1} \ROF }}
\lit{\spec{ \FORC{j = i \DOWNTO 1}  n_j.p \asg n_{j-1}; \ROF }}
\lit{\spec{ r \asg n_{i}; s \asg n_{i+1}; t \asg n_{i+2} }}

\lit{\reg{ s.p \asg r; }}
\lit{\reg{ r \asg s; s \asg t; t \asg t.p }}

\lit{\spec{ i \asg \rndl }}
\lit{\spec{ \FORC{j = 0 \TO \l-1}  n_j.p \asg n_{j+1} \ROF }}
\lit{\spec{ \FORC{j = i+1 \DOWNTO 1}  n_j.p \asg n_{j-1}; \ROF }}
\lit{\spec{ r \asg n_{i+1}; s \asg n_{i+2}; t \asg n_{i+3} }}

\lit{\reg{ i \asg i+1; }}

\lit{\spec{ i \asg \rndlp }}
\lit{\spec{ \FORC{j = 0 \TO \l-1}  n_j.p \asg n_{j+1} \ROF }}
\lit{\spec{ \FORC{j = i \DOWNTO 1}  n_j.p \asg n_{j-1}; \ROF }}
\lit{\spec{ r \asg n_{i}; s \asg n_{i+1}; t \asg n_{i+2} }}

\lio{\reg{\ttELIHW}}

\lio{\spec{i \asg \l-2}}
\lio{\spec{ r \asg n_{i}; s \asg n_{i+1}; t \asg n_{i+2} }}
\lio{\spec{ \FORC{j = 0 \TO \l-1}  n_j.p \asg n_{j+1} \ROF }}
\lio{\spec{ \FORC{j = i \DOWNTO 1}  n_j.p \asg n_{j-1}; \ROF }}

\lio{\reg{ s.p \asg r; }}

\lio{\spec{ \FORC{j = 0 \TO \l-1}  n_j.p \asg n_{j+1} \ROF }}
\lio{\spec{ \FORC{j = \l-1 \DOWNTO 1}  n_j.p \asg n_{j-1}; \ROF }}
\end{tabbing}
Now, as desired, the ``loop counter'' $i$ is no longer needed as a program variable, and can be
converted to an auxiliary (``ghost'') variable. The result is
\begin{tabbing}
mmm\=mmm\=mmm\= \kill  
\lio{\spec{ \FORC{j = 0 \TO \l-1}  n_j.p \asg n_{j+1} \ROF }}

\lio{\reg{ r \asg n_0; s \asg n_{1}; t \asg n_{2} }}

\lio{\spec{ i \asg 0 }}
\lio{\spec{ \FORC{j = 0 \TO \l-1}  n_j.p \asg n_{i+1} \ROF }}
\lio{\spec{ \FORC{j = i \DOWNTO 1}  n_j.p \asg n_{j-1}; \ROF }}
\lio{\spec{ r \asg n_{i}; s \asg n_{i+1}; t \asg n_{i+2} }}

\lio{\reg{\ttWHILEC{t \ne \nil}}}

\lit{\spec{ i \asg \rndl }}
\lit{\spec{ \FORC{j = 0 \TO \l-1}  n_j.p \asg n_{j+1} \ROF }}
\lit{\spec{ \FORC{j = i \DOWNTO 1}  n_j.p \asg n_{j-1}; \ROF }}
\lit{\spec{ r \asg n_{i}; s \asg n_{i+1}; t \asg n_{i+2} }}

\lit{\reg{ s.p \asg r; }}
\lit{\reg{ r \asg s; s \asg t; t \asg t.p }}


\lit{\spec{ i \asg \rndlp }}
\lit{\spec{ \FORC{j = 0 \TO \l-1}  n_j.p \asg n_{j+1} \ROF }}
\lit{\spec{ \FORC{j = i \DOWNTO 1}  n_j.p \asg n_{j-1}; \ROF }}
\lit{\spec{ r \asg n_{i}; s \asg n_{i+1}; t \asg n_{i+2} }}

\lio{\reg{\ttELIHW}}

\lio{\spec{i \asg \l-2}}
\lio{\spec{ r \asg n_{i}; s \asg n_{i+1}; t \asg n_{i+2} }}
\lio{\spec{ \FORC{j = 0 \TO \l-1}  n_j.p \asg n_{j+1} \ROF }}
\lio{\spec{ \FORC{j = i \DOWNTO 1}  n_j.p \asg n_{j-1}; \ROF }}

\lio{\reg{ s.p \asg r; }}

\lio{\spec{ \FORC{j = 0 \TO \l-1}  n_j.p \asg n_{j+1} \ROF }}
\lio{\spec{ \FORC{j = \l-1 \DOWNTO 1}  n_j.p \asg n_{j-1}; \ROF }}
\end{tabbing}
Upon termination, \reg{s} points to the head of the reversed list.
The post-program is quite pleasing: it constructs the list, and then immediately reverses it!

\section{Operational annotations with non-recursive procedures}

Let $pname$ be a non-recursive procedure with parameter passing by
value, and with body $pbody$.
Let $\b{a}$ denote a list of actual parameters, and
let $\b{f}$ denote a list of formal parameters.
An actual parameter is either an object identifier or an expression over
primitive types, and a formal parameter is either an
object identifier or a primitive-type identifier.
I handle non-recursive procedures in a similar manner to Hoare logic
verification rules for non-recursive procedures
\cite{Fa92}: formal parameters are replaced by actual parameters.
The most convenient expression of this principle is as an equivalence between a procedure call
and an instance of the procedure body with the appropriate assignment
of actuals to formals.
The equivalence rules for non-recursive procedure calls are:

%
\begin{equation}
  pname(\b{a}) \;\ev\; \b{f} \asg \b{a}; pbody \tag{\eqNonRV}
\end{equation}
\begin{equation}
r \asg pname(\b{a}) \;\ev\; \b{f} \asg \b{a}; pbody[ r \asg e/\RET e] \tag{\eqNonR}
\end{equation}

Rule \eqNonR is for when no value is returned (procedure type is void), or the returned
value is not saved by being assigned to a variable (return value is
``thrown away''). In this case, the form of the procedure call is
$pname(\b{a})$, \ie the procedure name followed by the actual
parameter list. Rule \eqNonR then assigns the formals to the actuals
(I assume a multiple assignment, with the obvious semantics) and
executes the procedure body.

Rule \eqNonR is for when the returned value is assigned to a variable.  In this case, the form of the procedure call is
$r \asg pname(\b{a})$, \ie an assignment statement in which the value
returned by the call  $pname(\b{a})$ is saved in variable $r$.
Rule \eqNonR then assigns the formals to the actuals
(I assume a multiple assignment, with the obvious semantics) and
executes a modified procedure body in which

I now give two examples of application of the above rules: selection
sort and linked list reversal.
First, consider the selection sort algorithm developed above,
packaged as a procedure:
\begin{tabbing}
  mmm\=mmm\=mmm\= \kill
\lio{\reg{void\ Procedure\ sort(int[]\ a,\ int\ n) \{}}
\lit{\reg{i \asg 0}}
\lit{\reg{ \ttWHILEC{i \ne n-1} }}
\lih{\reg{ J \asg \lrm(i:n-1); }}
\lih{\reg{ a[i] \swp a[J] }}
\lih{\reg{ i \asg i+1 }}
\lit{\reg{ \ttELIHW }}
\lion{\}}
\end{tabbing}
Now consider a call $sort(b,m)$ where $b$ is an array of length $m$.
By (\eqNonRV), I have

\begin{tabbing}
  mmm\=mmm\=mmm\= \kill
\lio{\reg{sort(b,m) \;\ev\;\{}} 
\lit{a := b; n := m;}
\lit{\reg{i \asg 0}}
\lit{\reg{ \ttWHILEC{i \ne n-1} }}
\lih{\reg{ J \asg \lrm(a, [0:i-1]); }}
\lih{\reg{ a[i] \swp a[J] }}
\lih{\reg{ i \asg i+1 }}
\lit{\reg{ \ttELIHW }}
\litn{\reg{\}}}
\end{tabbing}
This clearly has the correct effect. 
Now, consider the linked list reversal algorithm above, also
packaged as a procedure:
\begin{tabbing}
mm\=mmm\=mmm\= \kill  
\lio{\reg{Node\ Procedure\ listRev(Node\ h)}\ \{}
\lit{\reg{ r \asg h; s \asg h.next; t \asg h.next.next; }}
\lit{\reg{\ttWHILEC{t \ne \nil}}}
\lih{\reg{ s.p \asg r; }}
\lih{\reg{ r \asg s; s \asg t; t \asg t.p }}
\lit{\reg{\ttELIHW;}}
\lit{\reg{ s.p \asg r; }}
\lit{\reg{ \ttRET s }}
\lion{\reg{\}}}
\end{tabbing}
Now consider a call $v \asg listRev(\l)$ where $\l$ is the head of a linked list.
By (\eqNonR), I have:
\begin{tabbing}
mmm\=mmm\=mmm\= \kill  
\lio{\reg{v \asg listRev(\l) \;\ev \{}}
\lit{\reg{h \asg \l;}}
\lit{\reg{ r \asg h; s \asg h.next; t \asg h.next.next; }}
\lit{\reg{\ttWHILEC{t \ne \nil}}}
\lih{\reg{ s.p \asg r; }}
\lih{\reg{ r \asg s; s \asg t; t \asg t.p }}
\lit{\reg{\ttELIHW;}}
\lit{\reg{ s.p \asg r; }}
\lit{\reg{ v \asg s }}
\litn{\}}
\end{tabbing}
Again, this gives the correct effect for the procedure call $listRev(\l)$,
namely that $v$ points to the head of the reversed list. 
Note that both primitive and
reference types (as parameters) are handled correctly, and need not be
distinguished in the above rules.

\newcommand{\ctr}[1]{\ensuremath{\mathit{ct}{(#1)}}}
\newcommand{\cti}[1]{\ensuremath{\mathit{cti}{(#1)}}}
\newcommand{\tinsrt}[1]{\ensuremath{\mathit{insert}{(#1)}}}

\newcommand{\fset}{\ensuremath{\psi}}     
\newcommand{\aset}{\ensuremath{\vphi}}   
\newcommand{\fkey}{\ensuremath{k}}         
\newcommand{\akey}{\ensuremath{k}}   
\newcommand{\frt}{\ensuremath{T}}     
\newcommand{\fT}{\ensuremath{T}}     
\newcommand{\aT}{\ensuremath{T}}   
\newcommand{\art}{\ensuremath{T}}   

\section{Operational annotations with recursive procedures}

For recursive procedures, it is not sufficient to simply replace the
call by the body, since the body contains recursive instances of the
call. Clearly an inductive proof method is needed. What I use is an
inductive rule to establish the equivalence between a sequence of two
procedure calls and a third procedure call. These correspond,
respectively, to the pre-program, the program, and the post-program.
The method is as follows:
\bn

\item Let $\PrPR, \PrR, \PsPr$ be recursive procedures which give, respectively, the pre-program,
  program, and post-program

\item To establish $\PrPR; \PrR \ev \PsPr$ I proceed as follows:
  \bn
  \item Start with $\PrPR; \PrR$, replace the calls by their corresponding bodies, and then use
    sequence of equivalence-preserving transformations to bring the recursive calls of $\PrPR'; \PrR'$
    next to each other
  \item Use the inductive hypothesis for equivalence of the recursive calls $\PrPR'; \PrR' \ev \PsPr'$
    to replace $\PrPR'; \PrR'$ by $\PsPr'$
   \item Use more equivalence-preserving transformations to show that the resulting program is
     equivalent to $\PsPr$
  \en

\en  

The appropriate rule of inference is as follows:

\begin{equation}
\frac{
\begin{array}{l}
\PrPR'; \PrR' \ev \PsPr'\ \yld\  \PrPR; \PrR \ev \PsPr
\end{array}
}
   {\PrPR; \PrR \ev \PsPr}
\tag{\eqR}
\end{equation}

\noindent
where $\yld$ means ``is deducible from'', as usual.
This states that if we can prove  $\PrPR; \PrR \ev \PsPr$ (pre-program
followed by program is equivalent to post-program) by assuming $\PrPR'; \PrR' \ev \PsPr'$ 
(a recursive invocation of the pre-program followed by a recursive
invocation of the program is equivalent to a recursive invocation
of the post-program), then we can conclude, by induction on recursive
calls, that $\PrPR; \PrR \ev \PsPr$.

I illustrate this approach with an example which verifies the standard recursive algorithm for
inserting a node into a binary search tree (BST).
Each node $n$ in the tree consists of three
fields: $n.\key$ gives the key value for node $n$, $n.\l$ points to
the left child of $n$ (if any), and $n.r$ points to the right child of
$N$ (if any). The constructor $\NODE(v)$ returns a new node with key
value $v$ and null left and right child pointers.
I assume that all key values in the tree are unique.

The procedure $\tinsrt{t,v}$ gives the standard recursive algorithm
for insertion of key $v$ into a BST with root $t$.

\begin{tabbing}
m\=mmm\=mmm\= \kill  
\reg{\ul{\tinsrt{\fT,\fkey}}}::\+\\[1ex]
\reg{\ttIFC{\fT = \nil} \fT \asg {\ttNEW\ \NODE(\fkey)};}\\
\reg{\ttELSFC{\fkey < \frt.\key} \tinsrt{\fT.\l, \fkey};}\\
\reg{\ttELSE                              \tinsrt{\fT.r, \fkey}}; \`//\reg{\fkey > \fT.\key}\\
\end{tabbing}

To verify the correctness of $\tinsrt{\frt,\fkey}$, I define two recursive
procedures as follows.

The recursive procedure $\ctr{\fT,\fset}$ takes a set $\fset$ of key values,
and constructs a random binary search tree which contains exactly
these values, and sets $\fT$ to point to the root of this tree.
The statement $x \asg select\ in\ \fset$ selects a random
value in $\fset$ and assigns it to $x$.

The recursive procedure $\cti{\fT,\fset,\fkey}$ takes a set $\fset$ of key
values and a key value $k \not\in \fset$, constructs a random
binary search tree which contains exactly the values in $\fset$
together with the key $\fkey$, and where $\fkey$ is a leaf node, and
sets $\fT$ to point to the root of this tree.

\begin{tabbing}
m\=mmm\=mmm\= \kill  
\spec{\ul{\ctr{\fT,\fset}}}::\+\\[1ex]
\spec{\IFC{\fset = \emptyset} \fT \asg \nil;}\\
\spec{\ELSE}\+\\
\spec{x \asg select\ in\ \fset;}\\
\spec{\fset \asg \fset - x;}\\
\spec{\NODE\ \fT \asg \NEW\ \NODE(x);}\\
\spec{\ctr{ \fT.\l, \set{ y \stt y \in \fset \land y < x}};}\\
\spec{\ctr{ \fT.r, \set{ y \stt y \in \fset \land y > x}};}
\end{tabbing}

\begin{tabbing}
m\=mmm\=mmm\= \kill  
\spec{\ul{\cti{\fT,\fset,\fkey}}}::\+\\[1ex]
\spec{\IFC{\fset = \emptyset} \fT \asg {\NEW\ \NODE(k)}};\\
\spec{\ELSE}\+\\
\spec{x \asg select\ in\ \fset;}\\
\spec{\fset \asg \fset - x;}\\
\spec{\NODE\ n \asg \NEW\ \NODE(x);}\\

\spec{\IFC{\fkey < x}}\\
   \>\spec{\cti{\fT.\l, \set{y \stt y \in \fset \land y < x}, \fkey};}\\
   \> \spec{\ctr{\fT.r, \set{y \stt y \in \fset \land y > x}}}\\

\spec{\ELSE}\\
   \> \spec{\ctr{\fT.\l, \set{y \stt y \in \fset \land y < x}}}\\
   \> \spec{\cti{\fT.r, \set{y \stt y \in \fset\land y > x}, \fkey};}  \`//$\fkey > x$

\end{tabbing}

\noindent
I now verify
\[
\otp{\spec{\ctr{\aT, \aset}}}{\reg{\tinsrt{\aT, \akey}}}{\spec{\ctr{\aT, \aset \un \akey}}}.
\]
%
%
That is, the pre-program creates a random BST with key values in
$\aset$ and sets $\aT$ to the root, 
and the post-program creates a random BST with key values in
$\aset \un \akey$ and sets $\aT$ to the root.
Hence, the above operational triple states that the result of
$\art \asg \tinsrt{\aT, \akey}$ is to insert key value $\akey$ into the BST rooted at $\art$.
I first establish
\[
    \cti{\aT, \aset, \akey} \Pord \ctr{\aT, \aset \un \akey}. \tag{a}
\]
Intuitively, this follows since $\cti{\aT, \aset, \akey}$ constructs a BST with
key values in $\aset \un \akey$, and where $\akey$ is constrained to be a leaf
node, while \ctr{\aT, \aset \un \akey} constructs a BST with
key values in $\aset \un \akey$, with no constraint of where $\akey$ can
occur.
A formal proof proceeds by induction on the length of an arbitrary
execution $\pi$ of $\cti{\aT, \aset,\akey}$, which shows that $\pi$ is also a
possible execution of $\ctr{\aT, \aset \un \akey}$. Recall that the use of the
random selection statement means that there are, in general, many
possible executions for a given input. The details are straightforward
and are omitted.

In the sequel, I show that 
\[
  \otp{\spec{\ctr{\aT, \aset}}}{\reg{\tinsrt{\aT, \akey}}}{\spec{\cti{\aT, \aset, \akey}}} \tag{b}
\]
is valid. From (a,b) and (\postRlNb), I conclude that
\[
   \otp{\spec{\ctr{\aT, \aset}}}{\reg{\tinsrt{\aT, \akey}}}{\spec{\ctr{\aT, \aset \un \akey}}} 
 \]
is valid, as desired.
To establish (b), I show
\[
     \spec{\ctr{\aT, \aset}}; \reg{\tinsrt{\aT, \akey}} \ev \spec{\cti{\aT, \aset, \akey}}. \tag{c}
\]
from which (b) follows immediately by Definitions~\ref{def:progEquiv} and \ref{def:opTriple}.

I establish (c) by using induction on recursive calls.
I replace the above calls by the corresponding procedure bodies, and
then assume as inductive hypothesis (c) as applied to the recursive
calls within the bodies.
This is similar to the Hoare logic inference rule for partial correctness of recursive procedures \cite{Fa92}.

To be able to apply the inductive hypothesis, I take the sequential
composition 
$\spec{\ctr{\aT, \aset}}; \reg{\tinsrt{\aT, \akey}}$, 
replace each call by the corresponding procedure body, and then
I ``interleave'' the procedure bodies using commutativity of statements.
This enables me to bring the recursive calls to $ctr$ and to $tinsrt$
together, so that the inductive hypothesis can apply to their
sequential composition, which can then be replaced by the equivalent
recursive call to $cti$. This results in procedure body that
corresponds to a call of $cti$, which completes the equivalence proof.
I first replace $\spec{\ctr{\aT, \aset}}$ by the procedure body that
results from parameter binding:

\begin{tabbing}
m\=mmm\=mmm\= \+\kill  
\spec{\IFC{\aset = \emptyset} \aT \asg \nil;}\\
\spec{\ELSE}\+\\
\spec{x \asg select\ in\ \aset;}\\
\spec{\aset \asg \aset - x;}\\
\spec{\NODE\ \aT \asg \NEW\ \NODE(x);}\\
\spec{\ctr{\aT.\l, \set{ y \stt y \in \aset \land y < x}};}\\
\spec{\ctr{\aT.r, \set{ y \stt y \in \aset \land y > x}};}\\
\end{tabbing}

\noindent
I now take $\reg{\tinsrt{\aT, \akey}}$ and place it at the end
of both the if branch and the else branch, which clearly preserves
equivalence with $\spec{\ctr{\aT, \aset}}; \reg{\tinsrt{\aT, \akey}}$:

\begin{tabbing}
m\=mmm\=mmm\= \+\kill  
\spec{\IFC{\aset = \emptyset} \aT \asg \nil;} \reg{\tinsrt{\aT, \akey}}\\
\spec{\ELSE}\+\\
\spec{x \asg select\ in\ \aset;}\\
\spec{\aset \asg \aset - x;}\\
\spec{\NODE\ \aT \asg \NEW\ \NODE(x);}\\
\spec{\ctr{\aT.\l, \set{ y \stt y \in \aset \land y < x}};}\\
\spec{\ctr{\aT.r, \set{ y \stt y \in \aset \land y > x}};}\\
\reg{\tinsrt{\aT, \akey}}
\end{tabbing}

\noindent
I replace $\reg{\tinsrt{\aT, \akey}}$ by the procedure body that results from parameter binding:

\begin{tabbing}
m\=mmm\=mmm\= \+ \kill 
\reg{\ttIFC{\aT = \nil}\ \aT \asg {\ttNEW\ \NODE(\akey)}};\\
\reg{\ttELSFC{\akey < \aT.val} \tinsrt{\aT.\l, \akey}}\\
\reg{\ttELSE \tinsrt{\aT.r, \akey}}       \`//$\akey > \art.val$
\end{tabbing}

\noindent
In the if branch, I have $\spec{\aT \asg \nil}$, and so the body of
$\reg{\tinsrt{\aT, \akey}}$ simplifies to  

$\reg{\aT \asg {\ttNEW\ \NODE(\akey)}}$.

\noindent
In the else branch, I have \spec{\NODE\ \aT \asg \NEW\ \NODE(x);}, and
so the body of $\reg{\tinsrt{\aT, \akey}}$ simplifies to 

\begin{tabbing}
m\=mmm\=mmm\= \+ \kill  
\reg{\ttIFC{\akey < \aT.val}\ \tinsrt{\aT.\l, \akey}}\\
\reg{\ttELSE \tinsrt{\aT.r, \akey}}       \`//$\akey > \art.val$\\
\end{tabbing}

\noindent
I therefore now have

\begin{tabbing}
m\=mmm\=mmm\= \+\kill  
\spec{\IFC{\aset = \emptyset} \aT \asg \nil;}   \reg{\aT \asg {\ttNEW\ \NODE(\akey)}}      \\
\spec{\ELSE}\+\\
\spec{x \asg select\ in\ \aset;}\\
\spec{\aset \asg \aset - x;}\\
\spec{\NODE\ \aT \asg \NEW\ \NODE(x);}\\
\spec{\ctr{\aT.\l, \set{ y \stt y \in \aset \land y < x}};}\\
\spec{\ctr{\aT.r, \set{ y \stt y \in \aset \land y > x}};}\\
\reg{\ttIFC{\akey < \aT.val} \tinsrt{\aT.\l, \akey}}\\
\reg{\ttELSE \tinsrt{\aT.r, \akey}}       \`//$\akey > \art.val$
\end{tabbing}

\noindent
It is immediate that 
${\aT \asg \nil;}   \reg{\aT \asg {\ttNEW\ \NODE(\akey)}} \ev  \reg{\aT \asg {\ttNEW\ \NODE(\akey)}}$.
I also move 

\begin{tabbing}
m\=mmm\=mmm\= \+\kill  
\spec{\ctr{\aT.\l, \set{ y \stt y \in \aset \land y < x}};}\\
\spec{\ctr{\aT.r, \set{ y \stt y \in \aset \land y > x}};}
\end{tabbing}

\noindent
down into both branches of the following $\IF$ statement:

\begin{tabbing}
m\=mmm\=mmm\= \+\kill  
\spec{\IFC{\aset = \emptyset}}   \reg{\aT \asg {\ttNEW\ \NODE(\akey)}}      \\
\spec{\ELSE}\+\\
\spec{x \asg select\ in\ \aset;}\\
\spec{\aset \asg \aset - x;}\\
\spec{\NODE\ \aT \asg \NEW\ \NODE(x);}\\

\reg{\ttIFC{\akey < \aT.val}} \+\\
   \reg{\tinsrt{\aT.\l, \akey}}\\
   \spec{\ctr{\aT.\l, \set{ y \stt y \in \aset \land y < x}};}\\
   \spec{\ctr{\aT.r, \set{ y \stt y \in \aset \land y > x}};}\-\\

\reg{\ttELSE} \+     \`//$\akey > \art.val$\\
   \tinsrt{\aT.r, \akey}\\
   \spec{\ctr{\aT.\l, \set{ y \stt y \in \aset \land y < x}};}\\
   \spec{\ctr{\aT.r, \set{ y \stt y \in \aset \land y > x}};}
\end{tabbing}

\noindent
Since $\tinsrt{\aT.r, \akey}$ and $\spec{\ctr{\aT.\l, \set{ y \stt y \in \aset \land y < x}};}$ 
commute, the above is equivalent to

\begin{tabbing}
m\=mmm\=mmm\= \+\kill  
\spec{\IFC{\aset = \emptyset}}   \reg{\aT \asg {\ttNEW\ \NODE(\akey)}}      \\
\spec{\ELSE}\+\\
\spec{x \asg select\ in\ \aset;}\\
\spec{\aset \asg \aset - x;}\\
\spec{\NODE\ \aT \asg \NEW\ \NODE(x);}\\

\reg{\ttIFC{\akey < \aT.val}} \+\\
   \reg{\tinsrt{\aT.\l, \akey}}\\
   \spec{\ctr{\aT.\l, \set{ y \stt y \in \aset \land y < x}};}\\
   \spec{\ctr{\aT.r, \set{ y \stt y \in \aset \land y > x}};}\-\\

\reg{\ttELSE} \+     \`//$\akey > \art.val$\\
   \spec{\ctr{\aT.\l, \set{ y \stt y \in \aset \land y < x}};}\\
   \tinsrt{\aT.r, \akey;}\\
   \spec{\ctr{\aT.r, \set{ y \stt y \in \aset \land y > x}}}
\end{tabbing}

\noindent
I now apply the inductive hypothesis to conclude that

   \reg{\tinsrt{\aT.\l, \akey}};
   \spec{\ctr{\aT.\l, \set{ y \stt y \in \aset \land y < x}};}
   $\ev$
   \spec{\cti{\aT.\l, \set{ y \stt y \in \aset \land y < x}, \akey};}

\noindent
and that 

   \tinsrt{\aT.r, \akey};
   \spec{\ctr{\aT.r, \set{ y \stt y \in \aset \land y > x}}}
   $\ev$
   \spec{\cti{\aT.r, \set{ y \stt y \in \aset \land y > x}, \akey};}

\noindent
Making these substitutions results in:

\begin{tabbing}
m\=mmm\=mmm\= \+\kill  
\spec{\IFC{\aset = \emptyset}}   \reg{\aT \asg {\ttNEW\ \NODE(\akey)}}      \\
\spec{\ELSE}\+\\
\spec{x \asg select\ in\ \aset;}\\
\spec{\aset \asg \aset - x;}\\
\spec{\NODE\ \aT \asg \NEW\ \NODE(x);}\\

\reg{\ttIFC{\akey < \aT.val}} \+\\
   \spec{\cti{\aT.\l, \set{ y \stt y \in \aset \land y < x}, \akey};}\\
   \spec{\ctr{\aT.r, \set{ y \stt y \in \aset \land y > x}};}\-\\

\reg{\ttELSE} \+     \`//$\akey > \art.val$\\
   \spec{\ctr{\aT.\l, \set{ y \stt y \in \aset \land y < x}};}\\
   \spec{\cti{\aT.r, \set{ y \stt y \in \aset \land y > x}, \akey};}
\end{tabbing}

\noindent
and the above is seen to be the procedure body corresponding to the
call  \spec{\cti{\aT, \aset, \akey}}.

The above was formed by starting with $\spec{\ctr{\aT, \aset}}; \reg{\tinsrt{\aT, \akey}}$,
replacing the calls by the procedure bodies, and then performing a sequence of
equivalence-preserving transformations. Hence I conclude
$\spec{\ctr{\aT, \aset}}; \reg{\tinsrt{\aT, \akey}} \ev \spec{\cti{\aT, \aset, \akey}}$,
which is (c) above. This completes the proof.
For clarity, I have retained the formatting of pre-/post-program code in
red italics and of program code in black typewriter, even as I was
mixing these to establish the above equivalence.


\section{Related Work}

The use of assertions to verify programs was introduced by Floyd \cite{Fl67}
and Hoare \cite{Ho69}: a precondition $f$ expresses what can be assumed to
hold before execution of a program $P$, and a postcondition $g$ expresses what must hold
after the statement. The ``Hoare triple'' \htp{f}{P}{g} thus states
that if $f$ holds when execution of $P$ starts, then $g$ will hold
upon termination of $P$. If termination is not assumed, this is known
as \emph{partial correctness}, and if termination is assumed we have
\emph{total correctness}.
Both precondition and postcondition are
expressed as a formula of a suitable logic, \eg first order  logic.

Subsequently, Dijkstra introduced the \emph{weakest precondition
  predicate transformer} \cite{Dij75}: $wp(P,g)$ is the weakest
predicate $f$ whose truth before execution of $P$ guarantees $g$ afterwards, if $P$
terminates. He then used weakest preconditions to define a method for
formally deriving a program from a specification, expressed as a
precondition-postcondition pair.
Later, Hoare observed that the Hoare triple can be
expressed operationally, when he wrote ``$\htp{p}{q}{r} \df p;q < r$''
in \cite{Ho14}, but he does not seem to have developed this
observation into a proof system.

The formalization of both specifications and program correctness has lead
to a rich and extensive literature on program verification and
refinement. Hoare's original rules \cite{Ho69} were extended to deal
with  non-determinism, fair selection, and procedures \cite{Fa92}.
Separation logic \cite{Rey02} was introduced to deal with
pointer-based structures.

A large body of work deals with the notion of \emph{program
  refinement} \cite{BW98,Mo94}:
start with an initial artifact, which serves as a
specification, and gradually refine it into an executable and
efficient program.
This proceeds incrementally, in a sequence of
refinement steps, each of which preserves a ``refinement ordering'' relation $\PRef$, so
that we have $P_0 \PRef \ldots \PRef P_n$, where $P_0$ is the initial
specification and $P_n$ is the final program.
Morgan \cite{Mo94} starts with a pre-condition/post-condition
specification and refines it into an executable program using rules
that are similar in spirit to Dijkstra's weakest preconditions \cite{Dij75}.
Back and Wright \cite{BW98} use \emph{contracts}, which consist of
assertions (failure to hold causes a breach of the contract),
assumptions (failure to hold causes vacuous satisfaction of the
contract), and executable code.
As such, contracts subsume both pre-condition/post-condition pairs and
executable programs,
and so serve as an artifact for the seamless refinement of a 
pre-condition/post-condition specification into a program.


A related development has been the application of monads to
programming \cite{Mo89,Mo91}. A monad is an endofunctor $T$ over a
category $C$ together with a unit natural transformation from $1_C$
(the identity functor over $C$) to
$T$ and a multiplication natural transformation from $T^2$ to $T$. 
The Hoare state monad contains Hoare triples (precondition, program,
postcondition) \cite{J15},
and a computation maps an initial state to a pair consisting of a final
state and a returned value. The unit is the monadic operation
\textsf{return}, which 
lifts returned values into the state monad, and the multiplication is
the monadic operation
\textsf{bind}, which composes two computations, passing the resulting state
and returned value of the first computation to the second \cite{Sw09}.
The Dijkstra monad captures functions from postconditions
to preconditions \cite{J15,SHK16}. The \textsf{return} operation gives
the weakest precondition of a pure computation, and the 
\textsf{bind} operation gives the weakest precondition for a
composition of two computations.
 
Hoare logic and weakest preconditions are purely assertional proof methods.
Monads combine operational and assertional techniques, since they
provide operations which return the assertions that are used in the
correctness proofs. My approach is purely operational, since it uses
no assertions (formula in a suitable logic) but rather pre- and
post-programs instead. My approach thus represents the operational
endpont of the assertional--operational continuum, with Hoare
logic/weakest preconditions at the other (assertional) endpoint, and monads
somewhere in between.


\section{Conclusions}

I have presented a new method for verifying the correctness of sequential deterministic programs.
The method does not use assertions to specify correctness properties, but rather ``specification programs'',
which define a set of ``post states'', and can thus replace an assertion, which also defines a set of states, namely the states that satisfy it.
Since specification programs are not executed, they can be inefficient, and can refer to any mathematically well-defined quantity, \eg shortest path distances in a directed graph. In general, any formula of first/higher order logic can be referenced.

I illustrated my method on three examples: selection sort, Dijkstra's shortest path algorithm, and in-place list-reversal.
My approach has the following advantages, as illustrated by the examples:

\be

\item \emph{Code synthesis}: unwinding the outer loop of the pre-program and then trading it into the program can give initial code for a loop body of the program. This technique was illustrated in all three examples.

\item \emph{Trading}: trading gives great flexibility in developing both the program and the pre-program, as code can be freely moved between the
  program and the pre-program. This provides a tactic which is not available in logic-based verification methods such as
  Floyd-Hoare logic \cite{Fl67,Ho69} and  separation logic \cite{Rey02}.
  Trading was used in the sorting and list reversal examples.

\item \emph{Separation effect}: in the in-place list reversal example, a key requirement is that the reversed part of the list does not link
  around back to the non-reversed part. This requirement is easily expressed in my framework by the pre-program;
  in particular by how the pre-program constructs the list and then reverses part of it. Hence, using a pre-program which expresses, in an operational manner, the needed separation in pointer-based data structures, achieves the same effect as logic-based methods such as 
  separation logic \cite{Rey02}.

\item \emph{Refinement}: use a ``coarse'' and inefficient specification program and derive a more efficient program. Now iterate by using this program as a specification program to derive a still more efficient program, etc. My approach thus accommodates multi-level refinement.

\item \emph{Practical application}: it may be easier for developers to write specifications in code, a formalism that they are already well familiar with.

\ee


My approach is, to my knowledge, the first which uses \emph{purely} operational specifications to verify the correctness of sequential programs, as opposed to pre- and post-conditions and invariants in a logic such as Floyd-Hoare logic and separation logic, or axioms and signatures in
algebraic specifications \cite{Wir90}. The use of operational specifications is of course well-established
for the specification and verification of concurrent programs. The process-algebra approach \cite{CSP,CCS,PiCalc} starts with a specification written
in a process algebra formalism such as CSP, CCS, or the Pi-calculus, and then refines it into an implementation. Equivalence of the implementation and specification is established by showing a bisimulation \cite{CCS,Bisimulations} between the two. The I/O Automata approach \cite{IOA}
starts with a specification given as a single ``global property automaton'' and shows that a distributed/concurrent implementation respects the global property automaton by establishing a simulation relation \cite{Simulations} from the implementation to the specification.

My approach requires, in some cases, that one establish the equivalence of two programs \cite{CLRR16,PBD00,LR15}.  Any method for showing program equivalence works, since the equivalence is simply the hypothesis for the \eqSubRule rule.
The equivalence proofs in
this paper were informal and ``by inspection'', based on concepts such as the
commutativity of assignment statements that modify different variables/objects
%

Future work includes more examples and case studies, and in particular examples with pointer-based data structures.
I am also extending the operational annotations approach to the verification of concurrent programs.



\clearpage
\bibliographystyle{plain}
\bibliography{refs}

\begin{thebibliography}{10}

\bibitem{BW98}
Ralph{-}Johan Back and Joakim von Wright.
\newblock {\em Refinement Calculus - {A} Systematic Introduction}.
\newblock Graduate Texts in Computer Science. Springer, 1998.

\bibitem{CLRR16}
{\c{S}}tefan Ciob{\^{a}}c{\u{a}}, Dorel Lucanu, Vlad Rusu, and Grigore Rosu.
\newblock A language-independent proof system for full program equivalence.
\newblock {\em Formal Aspects Comput.}, 28(3):469--497, 2016.

\bibitem{CG99}
Roy~L. Crole and Andrew~D. Gordon.
\newblock Relating operational and denotational semantics for input/output
  effects.
\newblock {\em Math. Struct. Comput. Sci.}, 9(2):125--158, 1999.

\bibitem{Dij75}
Edsger~W. Dijkstra.
\newblock Guarded commands, nondeterminacy and formal derivation of programs.
\newblock {\em Commun. ACM}, 18(8):453–457, August 1975.

\bibitem{Fl67}
R~Floyd.
\newblock Assigning meanings to programs.
\newblock In {\em Mathematical Aspects of Computer Science. Proceedings of
  Symposium on Applied Mathematics}, pages 19--32. American Mathematical
  Society, 1967.

\bibitem{Fa92}
Nissim Francez.
\newblock {\em Program verification}.
\newblock International computer science series. Addison-Wesley, 1992.

\bibitem{CSP}
C.~A.~R. Hoare.
\newblock {\em Communicating Sequential Processes}.
\newblock Prentice-Hall, 1985.

\bibitem{LawsP}
C.~A.~R. Hoare, Ian~J. Hayes, Jifeng He, Carroll Morgan, A.~W. Roscoe, Jeff~W.
  Sanders, Ib~Holm S{\o}rensen, J.~Michael Spivey, and Bernard Sufrin.
\newblock Laws of programming.
\newblock {\em Commun. {ACM}}, 30(8):672--686, 1987.

\bibitem{Ho69}
C.A.R. Hoare.
\newblock An axiomatic basis for computer programming.
\newblock {\em Communications of the ACM}, 12(10):576--580, 583, 1969.

\bibitem{Ho14}
Tony Hoare.
\newblock Laws of programming: The algebraic unification of theories of
  concurrency.
\newblock In Paolo Baldan and Daniele Gorla, editors, {\em {CONCUR} 2014 -
  Concurrency Theory - 25th International Conference, {CONCUR} 2014, Rome,
  Italy, September 2-5, 2014. Proceedings}, volume 8704 of {\em Lecture Notes
  in Computer Science}, pages 1--6. Springer, 2014.

\bibitem{J15}
Bart Jacobs.
\newblock Dijkstra and hoare monads in monadic computation.
\newblock {\em Theor. Comput. Sci.}, 604:30--45, 2015.

\bibitem{LR15}
Dorel Lucanu and Vlad Rusu.
\newblock Program equivalence by circular reasoning.
\newblock {\em Formal Aspects Comput.}, 27(4):701--726, 2015.

\bibitem{IOA}
Nancy~A. Lynch and Mark~R. Tuttle.
\newblock An introduction to input/output automata.
\newblock {\em CWI-Quarterly}, 2(3):219--246, September 1989.
\newblock Centrum voor Wiskunde en Informatica, Amsterdam, The Netherlands.
  Technical Memo MIT/LCS/TM-373, Laboratory for Computer Science, Massachusetts
  Institute of Technology, Cambridge, MA 02139, November 1988.

\bibitem{Simulations}
Nancy~A. Lynch and Frits~W. Vaandrager.
\newblock Forward and backward simulations: I. untimed systems.
\newblock {\em Inf. Comput.}, 121(2):214--233, 1995.

\bibitem{CCS}
Robin Milner.
\newblock {\em A Calculus of Communicating Systems}, volume~92 of {\em Lecture
  Notes in Computer Science}.
\newblock Springer, 1980.

\bibitem{PiCalc}
Robin Milner.
\newblock {\em Communicating and mobile systems - the Pi-calculus}.
\newblock Cambridge University Press, 1999.

\bibitem{Mo89}
Eugenio Moggi.
\newblock Computational lambda-calculus and monads.
\newblock In {\em Proceedings of the Fourth Annual Symposium on Logic in
  Computer Science {(LICS} '89), Pacific Grove, California, USA, June 5-8,
  1989}, pages 14--23. {IEEE} Computer Society, 1989.

\bibitem{Mo91}
Eugenio Moggi.
\newblock Notions of computation and monads.
\newblock {\em Inf. Comput.}, 93(1):55--92, 1991.

\bibitem{Mo94}
Carroll Morgan.
\newblock {\em Programming from specifications, 2nd Edition}.
\newblock Prentice Hall International series in computer science. Prentice
  Hall, 1994.

\bibitem{Bisimulations}
David Michael~Ritchie Park.
\newblock Concurrency and automata on infinite sequences.
\newblock In Peter Deussen, editor, {\em Theoretical Computer Science, 5th
  GI-Conference, Karlsruhe, Germany, March 23-25, 1981, Proceedings}, volume
  104 of {\em Lecture Notes in Computer Science}, pages 167--183. Springer,
  1981.

\bibitem{PBD00}
Andrew~M. Pitts.
\newblock Operational semantics and program equivalence.
\newblock In Gilles Barthe, Peter Dybjer, Luis Pinto, and Jo{\~{a}}o Saraiva,
  editors, {\em Applied Semantics, International Summer School, {APPSEM} 2000,
  Caminha, Portugal, September 9-15, 2000, Advanced Lectures}, volume 2395 of
  {\em Lecture Notes in Computer Science}, pages 378--412. Springer, 2000.

\bibitem{PS93}
Andrew~M. Pitts and Ian David~Bede Stark.
\newblock Observable properties of higher order functions that dynamically
  create local names, or what's new?
\newblock In Andrzej~M. Borzyszkowski and Stefan Sokolowski, editors, {\em
  Mathematical Foundations of Computer Science 1993, 18th International
  Symposium, MFCS'93, Gdansk, Poland, August 30 - September 3, 1993,
  Proceedings}, volume 711 of {\em Lecture Notes in Computer Science}, pages
  122--141. Springer, 1993.

\bibitem{Pl04}
Gordon~D. Plotkin.
\newblock A structural approach to operational semantics.
\newblock {\em J. Log. Algebraic Methods Program.}, 60-61:17--139, 2004.

\bibitem{Rey02}
John~C. Reynolds.
\newblock Separation logic: A logic for shared mutable data structures.
\newblock In {\em Proceedings of the 17th Annual IEEE Symposium on Logic in
  Computer Science}, LICS '02, pages 55--74, Washington, DC, USA, 2002. IEEE
  Computer Society.

\bibitem{Sc14}
David~A. Schmidt.
\newblock Programming language semantics.
\newblock In Teofilo~F. Gonzalez, Jorge Diaz{-}Herrera, and Allen Tucker,
  editors, {\em Computing Handbook, Third Edition: Computer Science and
  Software Engineering}, pages 69: 1--19. {CRC} Press, 2014.

\bibitem{SHK16}
Nikhil Swamy, Catalin Hritcu, Chantal Keller, Aseem Rastogi, Antoine
  Delignat{-}Lavaud, Simon Forest, Karthikeyan Bhargavan, C{\'{e}}dric Fournet,
  Pierre{-}Yves Strub, Markulf Kohlweiss, Jean~Karim Zinzindohoue, and
  Santiago~Zanella B{\'{e}}guelin.
\newblock Dependent types and multi-monadic effects in {F}.
\newblock In Rastislav Bodik and Rupak Majumdar, editors, {\em Proceedings of
  the 43rd Annual {ACM} {SIGPLAN-SIGACT} Symposium on Principles of Programming
  Languages, {POPL} 2016, St. Petersburg, FL, USA, January 20 - 22, 2016},
  pages 256--270. {ACM}, 2016.

\bibitem{Sw09}
Wouter Swierstra.
\newblock A hoare logic for the state monad.
\newblock In Stefan Berghofer, Tobias Nipkow, Christian Urban, and Makarius
  Wenzel, editors, {\em Theorem Proving in Higher Order Logics}, pages
  440--451, Berlin, Heidelberg, 2009. Springer Berlin Heidelberg.

\bibitem{Wir90}
Martin Wirsing.
\newblock Algebraic specification.
\newblock In Jan van Leeuwen, editor, {\em Handbook of Theoretical Computer
  Science, Volume {B:} Formal Models and Semantics}, pages 675--788. Elsevier
  and {MIT} Press, 1990.

\end{thebibliography}

\end{document}